\documentclass[twoside,11pt]{article}

\usepackage{jmlr2e}

\usepackage{hyperref}
\usepackage{url}

\usepackage{caption}
\usepackage{subcaption}

\usepackage{graphicx} 
\usepackage{algorithm}
\usepackage{hyperref}
\usepackage{amsmath, amssymb}
\usepackage{color}

\usepackage{caption}
\usepackage{subcaption}
\usepackage{algorithm}

\usepackage{algpseudocode}

\newcommand{\beq}{\vspace{0mm}\begin{equation}}
\newcommand{\eeq}{\vspace{0mm}\end{equation}}
\newcommand{\beqs}{\vspace{0mm}\begin{eqnarray}}
\newcommand{\eeqs}{\vspace{0mm}\end{eqnarray}}
\newcommand{\barr}{\begin{array}}
\newcommand{\earr}{\end{array}}

\newcommand{\Xmat}[0]{{{\bf X}}}

\newcommand{\bv}[0]{{\boldsymbol{b}}}

\newcommand{\xv}{\boldsymbol{x}}

\newcommand{\zv}{\boldsymbol{z}}

\newcommand{\given}{\,|\,}

\newcommand{\betav}[0]{{\boldsymbol{\beta}}}

\newcommand{\E}{\mathbb{E}}

\newtheorem{thm}{Theorem} 
\newtheorem{cor}[thm]{Corollary}
\newtheorem{lem}[thm]{Lemma}

\usepackage{booktabs}
\usepackage{array}
\newcolumntype{L}[1]{>{\raggedright\let\newline\\\arraybackslash\hspace{0pt}}m{#1}}
\newcolumntype{C}[1]{>{\centering\let\newline\\\arraybackslash\hspace{0pt}}m{#1}}
\newcolumntype{R}[1]{>{\raggedleft\let\newline\\\arraybackslash\hspace{0pt}}m{#1}}

\usepackage{changes}

\newcolumntype{P}[1]{>{\centering\arraybackslash}p{#1}}

\usepackage{lastpage}
\jmlrheading{18}{2018}{1-\pageref{LastPage}}{7/17; Revised 12/17}{4/18}{17-409}
{AUTHORS' FULL NAMES}

\ShortHeadings{Permuted and Augmented Stick-Breaking 
Bayesian Multinomial Regression}{Zhang and Zhou}

\begin{document}

\title{Permuted and Augmented Stick-Breaking\\ 
Bayesian Multinomial Regression}

\author{\name Quan Zhang \email quan.zhang@mccombs.utexas.edu \\ \name Mingyuan Zhou 
 \email mingyuan.zhou@mccombs.utexas.edu\\
              \addr 
              Department of Information, Risk, and Operations Management\\
              McCombs School of Business\\
        The University of Texas at Austin\\
       Austin, TX 78712, USA}

\editor{David M. Blei}

\maketitle

\begin{abstract}
To model categorical response variables given their covariates, we propose a permuted and augmented stick-breaking (paSB) construction that one-to-one maps the observed categories to randomly permuted latent sticks. This new construction transforms multinomial regression into regression analysis of stick-specific binary random variables that are mutually independent given their covariate-dependent stick success probabilities, which are parameterized by the regression coefficients of their corresponding categories. The paSB construction allows transforming an arbitrary cross-entropy-loss binary classifier into a Bayesian multinomial one. Specifically, we parameterize the negative logarithms of the stick failure probabilities with a family of covariate-dependent softplus functions to construct nonparametric Bayesian multinomial softplus regression, and transform Bayesian support vector machine (SVM) into Bayesian multinomial SVM. These Bayesian multinomial regression models are not only capable of providing probability estimates, quantifying uncertainty, increasing robustness, and producing nonlinear classification decision boundaries, but also amenable to posterior simulation. Example results demonstrate their attractive properties and performance. 

\end{abstract}


\section{Introduction}
\label{intro}
Inferring the functional relationship between a categorical response variable and its covariates is a fundamental problem in physical and social sciences. To address this problem, it is common to use either multinomial logistic regression (MLR) 
\citep{mcfadden1973conditional,greene2003econometric,train2009discrete} 
or multinomial probit regression \citep{albert1993bayesian,mcculloch1994exact,mcculloch2000bayesian,imai2005bayesian},  both of which can be expressed as a latent-utility-maximization model that lets an individual make the decision by comparing its random utilities across all categories at once.
In this paper, 
we address the problem 
via a new stick-breaking construction of the multinomial distribution, which defines a one-to-one random mapping between the category and stick indices.
Rather than assuming an individual compares its random utilities across all categories at once, 
we assume an individual makes a sequence of stick-specific binary random decisions. 
The choice of the individual is 
the category mapped to the stick that is the first to choose ``1,'' or 
the category mapped to stick $S$ if all the first $S-1$ sticks choose ``0.''
 This framework transforms the problem of regression analysis of categorical variables into the problem of inferring the one-to-one mapping between the category and stick indices, and performing regression analysis of binary stick-specific random variables. 
 
Both 
MLR and the proposed 
stick-breaking models 
link a categorical response variable 
to its covariate-dependent probability 
parameters. 
While MLR is invariant to the permutation of category labels, given a fixed category-stick mapping, the proposed stick-breaking 
models purposely destruct that invariance. We are motivated to introduce this new framework for discrete choice modeling mainly to facilitate efficient Bayesian inference via data augmentation, introduce nonlinear decision boundaries, and relax a well-recogonized restrictive model assumption of MLR, as described below.

An important motivation is to extend efficient Bayesian inference available to binary 
regression to 
multinomial one. 
In the proposed stick-breaking models, 
the binary stick-specific random variables of an individual 
are 
conditionally independent given their stick-specific covariate-dependent probabilities.
Under this setting, one can solve a multinomial regression by solving conditionally independent binary ones. The only requirement is that the underlying  binary regression model 
uses the cross entropy loss. In other words, we require each stick-specific binary random variable to be linked via the Bernoulli distribution to its corresponding stick-specific covariate-dependent probability parameter.

Another important motivation is to improve the model capacity of MLR, which is a linear classifier 
in the sense that if the total number of categories is $S$, then MLR uses the intersection of $S-1$ linear hyperplanes to separate one class from the others. 
By choosing nonlinear binary regression models, we are able to enhance the capacities of the proposed stick-breaking models.
We are also motivated 
to relax the \emph{independence of irrelevant alternative} (IIA) assumption, an inherent property of MLR that requires the probability ratio of any two choices to be independent of the presence or characteristics of any other alternatives \citep{mcfadden1973conditional,greene2003econometric,train2009discrete}. By contrast,  
the proposed stick-breaking models make the probability ratio of two choices depend on 
other alternatives, as long as the two sticks that both choices are mapped to are not next to each other. 

In light of these considerations, we will first extend the softplus regressions recently proposed in \citet{SoftplusReg_2016}, a family of cross-entropy-loss binary classifiers that can introduce nonlinear 
decision boundaries and can recover 
logistic regression as a special case, 
to construct Bayesian multinomial softplus regressions (MSRs). 
We then consider a multinomial generalization of the widely used support vector machine (SVM) \citep{boser1992training,
cortes1995support,scholkopf1999advances,SVM},
a max-margin binary classifier that uses the hinge loss.
While there has been significant effort in extending binary SVMs into multinomial ones \citep{crammer2002algorithmic, lee2004multicategory,liu2011reinforced}, 
 the resulted extensions typically only provide the predictions of deterministic class labels. 
By contrast, 
we extend the Bayesian binary SVMs in \citet{sollich2002bayesian} and \citet{mallick2005bayesian} under the proposed framework to construct Bayesian multinomial SVMs (MSVMs), which naturally provide predictive class probabilities.

We will show that the proposed Bayesian MSRs 
 and MSVMs, which all generalize the stick-breaking construction to perform Bayesian multinomial regression, 
are not only capable of placing nonlinear decision boundaries between different categories, 
but also amenable to posterior simulation via data augmentation. Another attractive feature shared by all these proposed Bayesian algorithms is that they can 
not only predict class probabilities but also quantify model uncertainty. In addition, we  will show that robit regression, a robust cross-entropy-loss  binary classifier proposed in   \citet{liu2004robit}, can be extended into a robust Bayesian multinomial classifier under the proposed stick-breaking construction.

The remainder of the paper 
is organized as follows. In Section \ref{sec:stick_Mult} we briefly review MLR and discuss the restrictions of its stick-breaking construction. 
In Section \ref{sec:paSB} we propose the permuted and augmented stick breaking (paSB) 
to construct Bayesian multi-class classifiers, 
present  the inference, 
and show 
how the IIA assumption is relaxed. 
Under the paSB framework, we show how to transform softplus regressions 
 and support vector machines into Bayesian multinomial regression models in Sections \ref{sec:MSR} and \ref{sec:msvm}, 
 respectively.
We provide experimental results in Section \ref{eg} and conclude the paper in Section \ref{sec:conclude}. 

\section{Multinomial Logistic Regression and Stick Breaking}\label{sec:stick_Mult}

In this section we first briefly review multinomial logistic regression (MLR). We then use the stick-breaking construction to show how to generate 
a categorical random variable as a sequence of dependent binary variables, and further discuss a naive approach to transform binary logistic regression under stick breaking into multinomial regression. 
In the following discussion, we use $i\in\{1,\ldots,N\}$ to index the individual/observation, $s\in\{1,\ldots,S\}$ to index the choice/category, and the prime symbol to denote the transpose operation. 

\subsection{Multinomial Logistic Regression}

MLR that parameterizes the probability of each category given the covariates as
\beq\label{eq:MLR}
\textstyle P(y_i = s\given \xv_i,\{\betav_s\}_{1,S}) = p_{is}, ~p_{is} = {e^{\xv_i'\betav_s}}\big/{\big(\sum_{j=1}^S e^{\xv_i'\betav_{j}}\big)}
\eeq
is widely used, 
where 
$\xv_i\in\mathbb{R}^{P+1}$ consists of $x_{i1}=1$ and $P$ covariates, and $\betav_s\in\mathbb{R}^{P+1}$ consists of the regression coefficients for the $s$th category 
\citep{mccullagh1989generalized,albert1993bayesian,holmes2006bayesian}. 
Without loss of generality, one may choose category $S$ as the reference category by setting 
all the elements of $\betav_S$ as 0, making $e^{\xv_i'\betav_S}=1$ almost surely (a.s.). 
 For MLR, if data $i$ is assigned to the category with the largest 
 $p_{is}$, 
then one may consider that 
category $s$ resides 
 within a convex polytope \citep{polytope}, defined by
the set of solutions to $S-1$ inequalities as
$
\xv'(\betav_j-\betav_s)\le0$, where $j\in\{1,\ldots,s-1,s+1,\ldots,S\}$. 

Despite its popularity, MLR is a linear classifier in the sense that it uses the intersection of $S-1$ linear hyperplanes to separate one class from the others. As a classical discrete choice model in econometrics, it 
 makes the independence of irrelevant alternatives (IIA) assumption, implying that the unobserved factors for choice making 
are both uncorrelated and having the same variance across all alternatives \citep{mcfadden1973conditional,train2009discrete}. 
Moreover, while its log-likelihood is convex and there are efficient iterative algorithms to find 
the maximum likelihood or 
maximum a posteriori solutions of $\betav_s$, the absence of conjugate priors on $\betav_s$ makes it difficult to derive efficient Bayesian inference.
For Bayesian inference, \cite{polson2013bayesian} have introduced the P\'olya-Gamma data augmentation for logit models, and combined it with the data augmentation technique of \cite{holmes2006bayesian} for the multinomial likelihood to develop a Gibbs sampling algorithm for MLR. This algorithm, however, has to update $\betav_s$ one at a time while conditioning on all $\betav_{j}$ for $j\neq s$. Thus it may not only lead to slow convergence and mixing, especially when the number of categories $S$ is large, but also prevent us from parallelizing the sampling of $\{\betav_s\}_{1,S}$ within each 
MCMC iteration. 

\subsection{Stick Breaking}\label{sec:sbdef}
Suppose $y_i$ is
 a random variable drawn from a categorical distribution with a finite vector of probability parameters $(p_{i1},\ldots,p_{iS})$, where $S<\infty$, $p_{is}\ge0$, and $\sum_{s=1}^Sp_{is} =1$. Instead of directly using $y_i\sim\sum_{s=1}^S p_{is}\delta_s$, one may consider generating $y_i$ using the multinomial stick-breaking construction 
 that sequentially draws binary random variables 
 \beq\label{eq:StickBreaking}
 \textstyle{b_{is} \,\big{|}\, \{b_{ij}\}_{j<s}\,\sim\,\mbox{Bernoulli} \left[\left(1-\sum_{j<s}b_{ij}\right) \pi_{is}\right]}, ~~\displaystyle\pi_{is} = 
 \frac{p_{is}}{1-\sum_{j<s}p_{ij}}
 \eeq
 for $s=1,2,\ldots,S$. Note that $\pi_{iS}=1$ and $b_{iS} = 1- \sum_{j=1}^{S-1}b_{ij}$ by construction. Defining $y_i=s$ if and only if $b_{is}=1$ and $b_{ij}=0$ for all $j\neq s$, then one has a strick-breaking representation for the multinomial probability parameter as
 \beqs\label{eq:StickBreaking1}
 &P(y_i = s\given \{\pi_{is}\}_{1,S}) = P(b_{is}=1)\prod_{j\neq s}P(b_{ij}=0) = \pi_{is}\prod_{j<s}(1-\pi_{ij}), 
 \eeqs
 which, as expected, recovers $p_{is}$ by substituting the definitions of $\pi_{is}$ shown in \eqref{eq:StickBreaking}. 

The finite stick-breaking construction in \eqref{eq:StickBreaking1} can be further generalized to an infinite setting, as widely used in Bayesian nonparametrics \citep{hjort2010bayesian}. 
For example, the stick-break construction of \citet{sethuraman1994constructive} represents the length of the $k$th stick
using the product of $k$ stick-specific probabilities that are independent, and identically distributed (i.i.d.) beta random variables. 
It represents a size-biased random permutation of a Dirichlet process (DP) \citep{ferguson73} random draw, which includes countably infinite atoms whose weights sum to one. The stick-breaking construction of \citet{sethuraman1994constructive} has also been generalized to represent a draw from a random probability measure that is more general than the DP \citep{Pitman96somedevelopments,ishwaran2001gibbs,VBHDP}.

Related to this paper, one may further consider making the stick-specfic probabilities depend on the covariates \citep{KSBP,PSBP_Dunson,Lu_LSBP}. For example, the logistic stick-breaking process of \citet{Lu_LSBP} uses the product of $k$ covariate-dependent logistic functions to parameterize the probability of the $k$th stick.
To implement a stick-breaking process mixture model, truncated stick-breaking representations with a finite 
number of sticks are commonly used, with inference developed via both Gibbs sampling \citep{ishwaran2001gibbs,KSBP,PSBP} and variational approximation \citep{blei2006variational,kurihara2007collapsed,Lu_LSBP}. 

Another related work is the order-based dependent Dirichlet processes of \citet{DDP}, which use an ordered stick-breaking construction for mixture modeling, encouraging the data samples close to each other in the covariate space to share similar orders of the sticks and hence similar mixture weights. We will show that the proposed stick-breaking construction is distinct in that all data samples share the same category-stick mapping inferred from the data, with the category labels mapped to lower-indexed sticks subject to fewer geometric constraints on their decision boundaries.

\subsection{Logistic Stick Breaking}\label{LSB}
 
The stick-breaking construction parameterizes each $p_{is}$ with the product of $s$ probability parameters and links each $y_i$ with a unit-norm binary vector $(b_{i1},\ldots,b_{iS})$, where $b_{iy_i}=1$ and $b_{ij}=0$ a.s. if $j\neq y_i$. Following the logistic stick-breaking construction of \cite{Lu_LSBP}, one may represent $p_{is}$ with \eqref{eq:StickBreaking1} and parameterize the logit of each $\pi_{is}$ with a latent Gaussian variable $w_{is}$ as $\pi_{is} = {e^{w_{is}}}/{\sum_{j=1}^S e^{w_{ij}} }$.
To model observed or latent multinomial variables, a stick-breaking procedure, closely related to that of \citet{Lu_LSBP},
is used in \citet{khan2012stick} to transform the modeling of multinomial probability parameters into the modeling of the logits of binomial probability parameters using Gaussian latent variables. 
As shown in \citet{linderman2015dependent}, this procedure allows using the P\'olya-Gamma data augmentation, without requiring the assistance of the technique of \citet{holmes2006bayesian}, to construct Gibbs sampling that simultaneously updates all categories in each MCMC iteration, leading to improved performance over the one proposed in \citet{polson2013bayesian}. 

The simplification brought by the stick-breaking representation, which 
stochastically arranges its categories in decreasing order, comes with a clear change in that it removes the invariance of the multinomial distribution to label permutation. 
While the loss of invariance to label permutation may not pose a major issue for Bayesian mixture models inferred with MCMC \citep{jasra2005markov,kurihara2007collapsed}, it appears to be a major obstacle when applying stick breaking for multinomial regression, where the performance is often found to be sensitive to how the labels of the $S$ categories are ordered.
In particular, if one 
constructs a 
logistic stick breaking model 
by letting $\mbox{logit}(\pi_{is}) = w_{is} = \xv_i'\betav_s$, which means $\pi_{is} = (1+e^{-\xv_i'\betav_{s}})^{-1}$, then one has
\beq\textstyle
p_{is} = \big(1+e^{-\xv_i'\betav_s}\big)^{-1} \prod_{j<s} \big({1+e^{\xv_i'\betav_j}}\big)^{-1}, \notag
\eeq
which clearly tends to impose fewer geometric constraints on the classification decision boundaries of a category with a smaller $s$. For example, $p_{i1} = \big(1+e^{-\xv_i'\betav_1}\big)^{-1}$ is larger than $50\%$ if $\xv_i'\betav_1>0$ while $p_{i2} = \big(1+e^{\xv_i'\betav_1}\big)^{-1}\big(1+e^{-\xv_i'\betav_2}\big)^{-1}$ is possible to be larger than $50\%$ only if both $\xv_i'\betav_1<0$ and $\xv_i'\betav_2>0$. We will use an example to illustrate this type of geometric constraints in 
Section \ref{nonlinear_msr}. 

Under the logistic stick-breaking construction, not only could the performance be sensitive to how the $S$ different categories are ordered, but the imposed geometric constraints could also 
be overly restrictive even if the categories are appropriately ordered. Below we 
address the first issue by introducing a permuted and augmented stick-breaking representation for a multinomial model, and the second issue by adding the 
ability to model nonlinearity.

\section{Permuted and Augmented Stick Breaking} 
\label{sec:paSB}

To turn the seemingly undesirable sensitivity of the stick-breaking construction to label permutation into a favorable model property, when label asymmetry is desired, and mitigate performance degradation, when label symmetry is desired, we introduce a permuted and augmented stick-breaking (paSB) construction for a multinomial distribution, making it straightforward to extend an arbitrary binary classifier with cross entropy loss into a Bayesian multinomial one. The paSB construction infers a one-to-one mapping between the labels of the $S$ categories and the indices of the $S$ latent sticks, transforming the problem from modeling a multinomial random variable into modeling $S$ conditionally independent binary ones. 
It not only allows for parallel computation within each MCMC iteration, but also improves the mixing of MCMC in comparison to the one used in \citet{polson2013bayesian}, which updates one regression-coefficient vector conditioning on all the others, as will be shown in Section \ref{comparison}. 
 Note that the number of distinct one-to-one label-stick mappings is $S!$, which quickly becomes too large to exhaustively search for the best mapping as $S$ increases. Our experiments will show that the proposed MCMC algorithm can quickly escape from a purposely poorly initialized mapping and subsequently switch between many different mappings that all lead to similar performance, suggesting an 
effective search space that is considerably smaller than~$S!$.

\subsection{Category-Stick Mapping and Data Augmentation}
 \label{sec:MH}
The proposed 
paSB construction randomly maps a category to one and only one of the $S$ latent sticks and makes the augmented Bernoulli random variables $\{b_{is}\}_{1,S}$ conditionally independent to each other given $\{\pi_{is}\}_{1,S}$. Denote $\zv=(z_1,\ldots,z_S)$ as a permutation of $(1,\ldots,S)$, where $z_s\in\{1,\ldots,S\}$ is the index of the stick that category $s$ is mapped to. Given the label-stick mapping $\zv$, let us 
denote $p_{is}(\zv)$ as the multinomial probability of category~$s$, and $\pi_{iz_s}(\xv_i,\betav_s)$ as the covariate-dependent stick probability 
that is associated with the covariates of observation $i$ and the stick that category $s$ is mapped to.
For notational convenience, we will write $\pi_{iz_s}(\xv_i,\betav_s)$ as $\pi_{iz_s}$ and $\pi_{ij}(\xv_i,\betav_{s:z_s=j})$ as $\pi_{ij}$. We emphasize that 
here the $s$th regression-coefficient vector $\betav_s$ is always associated with both category $s$ and the corresponding stick probabily $\pi_{iz_s}$, 
a construction that will facilitate the inference of the label-stick mapping $\zv$. 
The following 
{Theorem} shows how to generate a categorical random variable of $S$ categories with a set of $S$ conditionally independent Bernoulli random variables. This is 
key to transforming the problem from solving multinomial regression into solving $S$ binary regressions independently. 

%

\begin{thm}\label{thm1}
Suppose $y_i\sim\sum_{s=1}^S p_{is}(\zv)\delta_{s}$, where $[p_{i1}(\zv),\ldots,p_{iS}(\zv)]$ is a multinomial probability vector whose elements are constructed as 
\begin{align}
p_{is}(\zv) = (\pi_{iz_s})^{\mathbf{1}(z_s\neq S)}\prod_{j<z_s}(1-\pi_{ij}),\textstyle 
\label{eq:p_is}
\end{align} 
 then $y_i$ can be equivalently generated under
the permuted and augmented stick-breaking (paSB) construction as
 \begin{align} 
 y_i \sim\sum_{s=1}^S &\left\{ \textstyle \left[\mathbf{1}(b_{iz_s}=1)\right]^{\mathbf{1}(z_s\neq S)}\prod_{j<z_s}\mathbf{1}(b_{ij}=0)\right\} \delta_s\,, \label{eq:AugSB}\\
 b_{ij} &\sim\emph{\mbox{Bernoulli}}(\pi_{ij}),~~j\in\{1,\ldots,S\}.\label{eq:paSB}
\end{align}
\end{thm}

Distinct from the conventional stick breaking in \eqref{eq:StickBreaking} that maps category $s$ to stick $s$ and makes $b_{is}$ depend on $b_{ij}, j=1,\ldots, s-1$, 
 under the new construction in 
\eqref{eq:AugSB}-\eqref{eq:paSB}, 
the $S$ categories are now randomly permuted and then one-to-one mapped to $S$ sticks, 
and the augmented binary random variables $\{b_{ij}\}_j$ become mutually independent given $\{\pi_{ij}\}_j$.
Given $y_i$, we still have $b_{ij}=0$ for $j<z_{y_i}$ and $b_{iz_{y_i}}=1$ a.s., but impose no restriction on any $b_{ij}$ for $j>z_{y_i}$, whose conditional posteriors given $y_i$ and $\pi_{ij}$ remain the same as their priors.
These changes are key to appropriately ordering the latent sticks, more flexibly parameterizing $\pi_{iz_s}$ and hence $p_{is}(\zv)$, and maintaining tractable inference. 


With paSB, the problem of inferring the functional relationship between the categorical response $y_i$ and the corresponding covariates $\xv_i$ is now transformed into the problem of modeling $S$ conditionally independent binary regressions as
$$b_{iz_s} \given \xv_i, \betav_s \sim\mbox{Bernoulli}[\pi_{iz_s}(\xv_i,\betav_s)],~i=1,\ldots,N,~s=1,\ldots,S.$$
Note that the only requirement for the binary regression model 
under
 paSB is that it uses the Bernoulli likelihood. In other words, it uses the cross entropy loss \citep{murphy2012machine} as
$$
-\sum_{i=1}^N\ln P(b_{iz_s} \given \xv_i, \betav_s) = \sum_{i=1}^N \big\{-b_{iz_s} \ln \pi_{iz_s}(\xv_i,\betav_s) - (1-b_{iz_s}) \ln [1-\pi_{iz_s}(\xv_i,\betav_s)]\big\}.
$$
A basic choice is paSB logistic regression that lets $$\pi_{iz_s}(\xv_i,\betav_s) = 1/(1+e^{-\xv_i\betav_s}),$$ which becomes the same as the logistic stick breaking construction described in Section \ref{LSB} if $z_s = s$ for all $s\in\{1,\ldots,S\}$. 
Another choice is paSB-robit regression that extends robit regression of \citet{liu2004robit}, a robust binary classifier using cross entropy loss, into a robust Bayesian multinomial classifier. 
In robit regression, observation $i$ 
is labeled as 1 if $\xv_i'\betav+\varepsilon_i>0$ and as 0 otherwise, where $\varepsilon_i$ are independently drawn from a $t$-distribution with $\kappa$ degrees of freedom, denoted as $\varepsilon_i\overset{iid}{\sim}t_{\kappa}$. Consequently, the conditional class probability function of robit regression is $P(y_i=1\given \xv_i,\betav)=F_\kappa(\xv_i'\betav)$, where $F_\kappa$ is the cumulative density function of $t_{\kappa}$. The robustness is attributed to the heavy-tail property of $F_\kappa(\xv_i'\betav)$, which, if $\kappa<7$,  imposes less penalty than the conditional class probability function of logistic regression does  on misclassified observations that are 
far from the decision boundary. 
Applying Theorem \ref{thm1}, the category probability of paSB-robit regression with $\kappa$ degrees of freedom  is shown in \eqref{eq:p_is}, where $\pi_{iz_s}(\xv_i,\betav_s) =  F_\kappa(\xv_i'\betav_s)$.
The paSB-robit regression provides a simple solution to robust multiclass classification; with $\{b_{ij}\}_{i,j}$ defined in Theorem \ref{thm1}, we run independent binary robit regressions using the Gibbs sampler  proposed in \citet{liu2004robit}.

In addition to paSB, 
we define permuted and augmented reverse stick breaking (parSB) 
{ in the following Corollary}.
\begin{cor}\label{cor_parsb}
Suppose $y_i\sim\sum_{s=1}^S p_{is}\delta_{s}$ and 
 \begin{align}
 p_{is}(\zv)= \textstyle (1-\pi_{iz_s})^{\mathbf{1}(z_s\neq S)}\prod_{j<z_s}\pi_{ij}\,, \notag
\end{align} 
then $y_i$ can also be generated under the permuted and augmented reverse stick-breaking (parSB) representation as
 \begin{align}\vspace{-3mm}\label{eq:AugSB1_y}
 y_i \sim\sum_{s=1}^S& \left\{ \textstyle \left[\mathbf{1}(b_{iz_s}=0)\right]^{\mathbf{1}(z_s\neq S)}\prod_{j<z_s}\mathbf{1}(b_{ij}=1)\right\} \delta_s\,,\\
 \notag
 b_{ij} &\sim\emph{\mbox{Bernoulli}}(\pi_{ij}),~~j\in\{1,\ldots,S\}.
\end{align}
\end{cor}
Generally speaking, if $\pi_{iz_s}(\xv_i,-\betav_s) = 1-\pi_{iz_s}(\xv_i,\betav_s)$, which is the case for logistic stick breaking 
and robit stick breaking, where $\pi_{iz_s}$ are defined as  $(1+e^{-\xv_i'\betav_{s}})^{-1}$ and  $F_\kappa(\xv_i'\betav_s)$, respectively, 
and Bayesian multinomial SVMs to be discussed in Section \ref{sec:msvm}, then there is no need to introduce parSB as an addition to paSB. 
Otherwise, there are potential benefits, such as for softplus regressions to be introduced in Section~\ref{sec:MSR}, to 
combine parSB with paSB.

\subsection{Inference of Stick Variables and Category-Stick Mapping}\label{mh}
Below we first describe Gibbs sampling for the augmented stick variables $\{b_{ij}\}_{1,S}$, 
and then introduce a Metropolis-Hastings (MH) step to infer the category-stick mapping~$\zv$. 
Given the category label $y_i$, stick probability $\pi_{ij}$, and $\zv$, 
we sample $b_{ij}$ as
\beq
(b_{ij}\given y_i, \pi_{ij},\zv) \sim \mathbf{1}(j=z_{y_i}) + \mathbf{1}(j>z_{y_i}) \mbox{Bernoulli}(\pi_{ij}), \notag 
\eeq
for $j=1,\ldots,S-1$, and let 
\beq
b_{iS} = \mathbf{1}(z_{y_i}=S). \notag 
\eeq
This means we let $b_{ij}=0$ if ${j}<z_{y_i}$, $b_{ij}=1$ if $j=z_{y_i}$, draw $b_{ij}$ from $\mbox{Bernoulli}(\pi_{ij})$ if $z_{y_i}<j<S$, and let $b_{iS}=1$ if and only if $z_{y_i}=S$.
Note that  stick $S$ is used as a reference stick and 
$\pi_{iS}$ is not used in defining  $p_{is}(\zv)$  in 
\eqref{eq:p_is}.
Despite 
having no impact on computing $\{p_{is}\}_{1,S}$, we infer $\pi_{iS}$ ($i.e.$, sample the regression-coefficient vector $\betav_{s':z_{s'}=S}$) under the likelihood $\prod_{i=1}^N\mbox{Bernoulli}(b_{iS};\pi_{iS})$ 
and use it in a Metropolis-Hastings step, as described in (\ref{eq:MH}) shown 
below, 
to decide whether to switch the mappings of two different categories, if one of which is mapped to the reference stick $S$. 
Once we have an MCMC sample of $\{b_{ij}\}_{1,S}$, we then essentially solve independently $S$ binary classification problems, the $j$th of which can be expressed as
$b_{ij} \given \xv_i, \betav_{s:z_s=j} \sim\mbox{Bernoulli}[\pi_{ij}(\xv_i,\betav_{s:z_s=j})].$

Analogously, for parSB, $\{b_{ij}\}_{1,S}$ can be sampled as 
$
(b_{ij}\given y_i, \pi_{ij},\zv) \sim \mathbf{1}(j<z_{y_i}) + \mathbf{1}(j>z_{y_i}) \mbox{Bernoulli}(\pi_{ij}) 
$ 
for $j=1,\ldots,S-1$, and 
$ 
b_{iS} =1-\mathbf{1}(z_{y_i}=S), 
$ 
which means we let $b_{ij}=1$ if ${j}<z_{y_i}$, let $b_{ij}=0$ if $j=z_{y_i}$, draw $b_{ij}$ from $\mbox{Bernoulli}(\pi_{ij})$ if $z_{y_i}<j<S$, and let $b_{iS}=0$ if and only if $z_{y_i}=S$.

Since stick-breaking multinomial classification is not invariant to the permutation of its class labels, it may perform substantially worse than it could be if the inherent geometric constraints implied by the current ordering of the labels make it difficult to adapt the decision boundaries to the data.
Our solution to this problem is to infer the one-to-one mapping between the category labels and stick indices from the data.
 We construct a Metropolis-Hastings (MH) step within each Gibbs sampling iteration, with a proposal of switching two sticks that categories $c$ and $c^\prime$, $1\le c<c^\prime\le S$, are mapped to, by changing the current category-stick one-to-one mapping from $\zv=(z_1,\ldots,z_c,\ldots,z_{c^\prime},\ldots,z_S)$ to $\zv'=(z'_1,\ldots,z'_S):=(z_1,\ldots,z_{c^\prime},\ldots,z_c,\ldots,z_S)$.
Assuming a uniform prior on $\zv$ and proposing $(c,c^\prime)$ uniformly at random from one of the $\binom{S}{2}=
S(S-1)/2
$ possibilities,
we would accept the proposal with probability 
\beq\small \label{eq:MH}
\!\!\!
\min\left\{ \prod_{i}\frac{ \prod_{s=1}^S[p_{is}(\zv')]^{\mathbf{1}(y_i=s)}}{ \prod_{s=1}^S[p_{is}(\zv)]^{\mathbf{1}(y_i=s)}},~1\right\} = 
\min\left\{\prod_{i}\frac{ 
 \prod_{s=1}^S \left[(\pi_{iz'_s})^{\mathbf{1}(z'_s\neq S)}\prod_{j<z'_s}(1-\pi_{ij})\right]^{\mathbf{1}(y_i=s)}}{ 
\prod_{s=1}^S \left[(\pi_{iz_s})^{\mathbf{1}(z_s\neq S)}\prod_{j<z_s}(1-\pi_{ij})\right]^{\mathbf{1}(y_i=s)}}
,~1\right\}.
\eeq 
%
%

\subsection{Sequential Decision Making}

Random utility models, including both the logit and probit models as special examples, are widely used to infer the functional relationship between a categorical response variable and its covariates. For discrete choice analysis in econometrics \citep{hanemann1984discrete,greene2003econometric,train2009discrete}, these models %
assume that among a set of $S$ alternatives, an individual makes the choice that maximizes his/her utility $U_{is} = V_{is}+\varepsilon_{is}$, where 
$V_{is}$ and $\varepsilon_{is}$ represent the observable and unobservable parts of $U_{is}$, respectively. 
If $V_{is}$ is set as 
$V_{is} = \xv_i'\betav_s$, then marginalizing out $\boldsymbol \varepsilon_{i}=(\varepsilon_{i1},\ldots,\varepsilon_{iS})'$ leads to 
MLR if all $\varepsilon_{is}$ follow the extreme value distribution 
\citep{mcfadden1973conditional,greene2003econometric,train2009discrete}, 
and multinomial probit regression if 
all $\boldsymbol \varepsilon_{i}$ 
 follow
a multivariate normal distribution \citep{albert1993bayesian,mcculloch1994exact,mcculloch2000bayesian,imai2005bayesian}. 


Instead of examining the utilities of all choices before making the decision, 
the paSB construction is characterized by a sequential decision making process, 
described as follows. In step one, an individual decides whether to  select the choice mapped to stick 1, or to select a choice among the remaining alternatives, $i.e.$, choices $\{s:z_s\in\{2,\ldots,S\}\}$. 
If the individual selects the choice mapped to stick $1$, then the sequential process is terminated. Otherwise this choice  is eliminated and the individual proceeds to step two, in which he/she would follow the same procedure to either select the choice mapped to stick $2$ or proceed to the next step to select a choice among the remaining alternatives, $i.e.$, choices $\{s:z_s\in\{3,\ldots,S\}\}$. 
The individual, reconsidering none of the eliminated choices, will keep making a \emph{one-vs-remaining} decision at each step until the termination of the sequential decision making process. 

This unique sequential decision making procedure relaxes the independence of irrelevant alternatives (IIA) assumption, 
as described in the following Lemma. 

\begin{lem}\label{lem_iia} Under the paSB construction, the probability ratio of two choices are influenced by the success probabilities of the sticks that lie between these two choices' corresponding sticks. In other words, the probability ratio of two choices will be influenced by some other choices if they are not mapped to adjacent sticks. 
\end{lem}

As in Lemma \ref{lem_iia},  the paSB construction could adjust how two choices' probability ratio depends on the other alternatives
 by controlling the distance between the two sticks that they are mapped to, and hence provide a unique way to relax the IIA assumption. 
While the widely used MLR can be considered as a random-utility-maximization model with the IIA assumption, the paSB multinomial logistic model performs sequential random utility maximization 
that relaxes this assumption, 
as described in Lemma \ref{thm_latent_u} in the Appendix. 

\section{Bayesian Multinomial Softplus Regression}\label{sec:MSR}

Logistic regression is a cross-entropy-loss binary classifier that can be straightforwardly extended to paSB multinomial logistic regression (paSB-MLR). However, it is a linear classifier that uses a single hyperplane to separate one class from the other. To introduce nonlinear classification decision boundaries, we consider extending softplus regression of \citet{SoftplusReg_2016}, a multi-hyperplane binary classifier that uses the cross entropy loss, into multinomial softplus regression (MSR) under paSB. 

Softplus regression uses the interaction of multiple hyperplanes to construct a union of convex-polytope-like confined spaces to enclose the data labeled as ``1,'' which are hence separated from the data labeled as ``0''. 
It is constructed under a Bernoulli-Poisson link \citep{EPM_AISTATS2015} that thresholds at one a latent Poisson count, with the distribution of the Poisson rate defined as the convolution of the probability density functions of $K$ experts, each of which corresponds to the stack of $T$ gamma distributions with covariate-dependent scale parameters. The number of experts $K$ and the number of layers $T$ can be considered as the two model parameters that determine the nonlinear capacity of the model. More specifically, for expert $k$, denoting  $r_k$ as its weight and $\betav_k^{t+1}$ as its $t$th regression-coefficient vector, the conditional class probability can be expressed~as
\begin{align}
&P(y_i=1\given \xv_i,\{r_k,\{\betav_k^{(t+1)}\}_{1,T}\}_{1,K}) = 1 -\prod_{k=1}^K (1-p_{ik}), \notag\\
& p_{ik} 
= 1-\left(\!1\!+\!e^{\xv_i'\betav_{k}^{({T}\!+\!1)}}\ln\!\left\{1\!+\!e^{\xv_i'\betav_{k}^{({T})}}\ln\left[1\!+\!\ldots
\ln\left(1\!+\!e^{\xv_i'\betav_{k}^{(2)}}\right)\right]\right\}\right)^{-r_{k}} ;\notag
\end{align}
when $K=T=1$, the conditional class probability  reduces to
\beq
P(y_i=1\given \xv_i,r,\betav) = 1-\left(\frac{1}{1+e^{\xv_i'\betav}}\right)^r; \notag
\eeq
and when $K=T=r=1$, it becomes the same as that of binary logistic regression.
Note that a gamma process, a random draw from which is expressed as  $G=\sum_{k=1}^\infty r_k\delta_{\{\betav_k^{t+1}\}_{1,T}}$, can be used to support a potentially countably infinite number of experts for softplus regression. For this reason, one can set $K$ as large as permitted by computation and relies on the gamma process's inherent shrinkage mechanism to turn off unneeded model capacity (not all $K$ experts will be used if $K$ is set to be sufficiently large).



\subsection{paSB and parSB Extensions of Softplus Regressions}
We first follow \citet{SoftplusReg_2016} to define 
\beq \varsigma(x_1,\ldots,x_t) =
\ln\left(1+e^{x_t}\ln\left\{1+e^{x_{t-1}}\ln\left[1+\ldots
\ln\left(1+e^{x_1}\right)\right]\right\}\right) \notag 
 \eeq
 as the stack-softplus function. Note that if $t=1$, the stack-softplus function reduces to softplus function $\varsigma(x) = \ln(1+e^x)$, which is often considered as a smoothed version of the rectifier function, expressed  as $\mbox{rectifier}(x)=\max(0,x)$, that has become the dominant 
 nonlinear activation function for deep neural networks \citep{nair2010rectified, glorot2011deep, krizhevsky2012imagenet, lecun2015deep}. We then parameterize $\lambda_{iz_s} = -\ln(1-\pi_{iz_s})$, the negative logarithms of the failure probabilities of the stick that category $s$ is mapped to, as
\beqs\label{eq:lambda_is}
\lambda_{iz_s}=
\sum_{k=1}^\infty r_{sk}\, \varsigma\big(\xv'\betav_{sk}^{(2)},\ldots,\xv'\betav_{sk}^{(T+1)}\big),
\eeqs
where the countably infinite atoms $
(\betav_{sk}^{(2)},\ldots,\betav_{sk}^{(T+1)})
$ and their weights $\{r_{sk}\}_k$ constitute a draw from a gamma process $G_s\sim\mbox{GaP}(G_0,1/c_s)$ \citep{ferguson73}, with $G_0$ as a finite and continuous base distribution over a complete separable metric space $\Omega$ and $1/c_s$ 
as a scale parameter. In other words, we let $b_{iz_s}\sim\mbox{Bernoulli}(\pi_{iz_s})$ or 
\begin{align}
b_{iz_s}\sim{\mbox{Bernoulli}}\left[
1 -\prod_{k=1}^\infty \left(\!1\!+\!e^{\xv_i'\betav_{sk}^{({T}\!+\!1)}}\ln\!\left\{1\!+\!e^{\xv_i'\betav_{sk}^{({T})}}\ln\left[1\!+\!\ldots
\ln\left(1\!+\!e^{\xv_i'\betav_{sk}^{(2)}}\right)\right]\right\}\right)^{-r_{sk}} ~\!\right ]. \label{eq:b_is}
\end{align} 

As shown in Theorem 10 of \cite{SoftplusReg_2016}, $b_{iz_s}$
can be equivalently generated from
a hierarchical model that convolves countably infinite stacked gamma distributions, with covariate-dependent scales, as 
\begin{align}\small
&~~~~~\theta^{({T})}_{isk}\sim{\mbox{Gamma}}\left(r_{sk},e^{\xv_i'\betav^{({T}+1)}_{sk}}\right)
,\notag\\
&~~~~~~~~~~~~~~~~\ldots\notag\\
&~~~~~\theta^{(t)}_{isk}\sim{\mbox{Gamma}}\left(\theta^{(t+1)}_{isk},e^{\xv_i'\betav^{(t+1)}_{sk}}\right), \notag\\
&~~~~~~~~~~~~~~~~\ldots\notag\\
& 
~~~~~\theta^{(1)}_{isk}\sim{\mbox{Gamma}}\left(\theta^{(2)}_{isk},e^{\xv_i'\betav^{(2)}_{sk}}\right), \notag\\
b_{iz_s} = \mathbf{1}&(m_{is}\ge 1),~m_{is}
=\sum_{k=1}^\infty m^{(1)}_{isk},~m^{(1)}_{isk} \sim {\mbox{Pois}}(\theta^{(1)}_{isk}),
\label{eq:DICLR_model}
\end{align}
the marginalization of whose latent variables lead to \eqref{eq:b_is}.
Note the gamma distribution $\theta\sim\mbox{Gamma}(r,1/c)$ is defined such that $\E[\theta]=r/c$ and $\mbox{var}[\theta]=r/c^2$, and  the hierarchical structure in \eqref{eq:DICLR_model} can also be related to the augmentable gamma belief network proposed in \citet{GBN}. 
We consider the combination of \eqref{eq:DICLR_model} and either paSB in 
\eqref{eq:AugSB} or parSB in 
{\eqref{eq:AugSB1_y}} as the Bayesian nonparametric 
hierarchical model for multinomial softplus regression (MSR) that is defined below. 

 \begin{definition}[Multinomial Softplus Regression] 
With a draw from a gamma process for each category
that consists of countably infinite atoms $\betav_{sk}^{(2:{T}+1)}$ with weights $r_{sk}>0$, where $\betav_{sk}^{(t)}\in\mathbb{R}^{P+1}$, 
given the covariate vector $\xv_i$ and category-stick mapping 
 $\zv$, 
MSR 
 parameterizes $p_{is}$, the multinomial probability of category $s$, under the paSB construction as
 \begin{align}\vspace{-3mm}   
&p_{is}(\zv) =\textstyle \left[
1 -\prod_{k=1}^\infty \left(\!1\!+\!e^{\xv_i'\betav_{sk}^{({T}\!+\!1)}}\ln\!\left\{1\!+\!e^{\xv_i'\betav_{sk}^{({T})}}\ln\left[1\!+\!\ldots
\ln\left(1\!+\!e^{\xv_i'\betav_{sk}^{(2)}}\right)\right]\right\}\right)^{-r_{sk}} ~\!\right ]^{\mathbf{1}(z_s\neq S)}\notag\\
&\textstyle\times \prod_{j:z_j<z_s} \left[
\prod_{k=1}^\infty \left(\!1\!+\!e^{\xv_i'\betav_{jk}^{({T}\!+\!1)}}\ln\!\left\{1\!+\!e^{\xv_i'\betav_{jk}^{({T})}}\ln\left[1\!+\!\ldots
\ln\left(1\!+\!e^{\xv_i'\betav_{jk}^{(2)}}\right)\right]\right\}\right)^{-r_{jk}} ~\!\right ],\notag
\end{align}
and 
 parameterizes $p_{is} $ under the parSB construction as
 \begin{align}\vspace{-3mm}  
&p_{is}(\zv) =\textstyle \left[
\prod_{k=1}^\infty \left(\!1\!+\!e^{\xv_i'\betav_{sk}^{({T}\!+\!1)}}\ln\!\left\{1\!+\!e^{\xv_i'\betav_{sk}^{({T})}}\ln\left[1\!+\!\ldots
\ln\left(1\!+\!e^{\xv_i'\betav_{sk}^{(2)}}\right)\right]\right\}\right)^{-r_{sk}} ~\!\right ]^{\mathbf{1}(z_s\neq S)}\notag\\
&\!\!\textstyle\times \prod_{j:z_j<z_s} \left[1-
\prod_{k=1}^\infty \left(\!1\!+\!e^{\xv_i'\betav_{jk}^{({T}\!+\!1)}}\ln\!\left\{1\!+\!e^{\xv_i'\betav_{jk}^{({T})}}\ln\left[1\!+\!\ldots
\ln\left(1\!+\!e^{\xv_i'\betav_{jk}^{(2)}}\right)\right]\right\}\right)^{-r_{jk}} ~\!\right ].\notag
\end{align}
\end{definition}

For the convenience of implementation, we truncate the number of atoms of the gamma process at $K$ by choosing a discrete base measure for each category as $G_{s0} =\sum_{k=1}^K \frac{\gamma_{s0}}{K} \delta_{\betav_{sk}^{(2:{T}+1)}}$, under which we have
$r_{sk}\sim\mbox{Gamma}(\gamma_{s0}/K,1/c_{s0})$ as the prior distribution for the weight of expert $k$ in category~$s$. For each category, we expect only some of its $K$ experts to have non-negligible weights if $K$ is set large enough, and we may use $\sum_{k}\mathbf{1}\big(\sum_{i}m^{(1)}_{isk} >0\big)$, where $m^{(1)}_{isk}$ is defined in \eqref{eq:DICLR_model}, to measure the number of active experts inferred from the data. 

\subsection{Geometric Constraints for MSR}\label{sec:geometric}

Since by definition we have
$\textstyle p_{is}(\zv) = \pi_{iz_s}\big({1-\sum_{j<s}p_{is}(\zv)}\big)=\pi_{iz_s}\prod_{j<z_s}(1-\pi_{ij})$ in MSR, 
it is clear that if $\pi_{ij}$ for all $j<z_s$ are small and $\pi_{iz_s}$ is the first one to have a large probability value close to one, $y_i$ will be likely assigned to category $s$ regardless of how large the values of $\{\pi_{ij}\}_{j>z_s}$ are. 
To motivate the use of 
the seemingly over-parameterized sum-stack-softplus function 
 in \eqref{eq:lambda_is},
we first consider the simplest case of $K=T=1$.
Without loss of generality, let us assume that the category-stick mapping is fixed at $\zv=(1,\ldots,S)$.
\begin{lem}\label{lem:K1T1}
For paSB-MSR with $K=T=1$ and $\zv=(1,\ldots,S)$, the set of solutions to $p_{is}(\zv)>p_0$ in the covariate space 
are bounded by a convex polytope defined by the intersection of $s$ linear hyperplanes. 
\end{lem}
Note that the binary softplus regression with $K=T=1$ is closely related to logistic regression, and reduces to logistic regression if $r=1$ 
\citep{SoftplusReg_2016}. With Lemma \ref{lem:K1T1}, it is clear that even if an optimal category-stick mapping $\zv$ is provided, paSB-MSR with $K=T=1$ may still clearly underperform MLR. This is because category $s$ uses a single hyperplane to separate itself from the remaining $S-s$ categories, and hence uses the interaction of at most $s$ hyperplanes to separate itself from the other $S-1$ categories. 
By contrast, MLR uses a convex polytope bounded by at most $S-1$ hyperplanes for each of the $S$ categories.

When $K>1$ and/or $T>1$, an exact theoretical analysis 
is beyond the scope of this paper. 
 Instead we provide some qualitative analysis by borrowing related geometric-constraint analysis for softplus regressions in \cite{SoftplusReg_2016}. 
Note that Equation \eqref{eq:b_is} indicates that a noisy-or model  \citep{pearl2014probabilistic,srinivas1993generalization},  commonly appearing in causal inference, is used at each step of the sequential one-vs-remaining decision process; at each step, the binary outcome 
of an observation 
is attributed to the disjunctive interaction of many possible hidden causes. 
Roughly speaking, 
to enclose category $s$ to separate it from the remaining $S-s$ categories in the covariate space,
paSB-MSR with $K>1$ and $T=1$ uses the complement of a convex-polytope-bounded space, 
paSB-MSR with $K=1$ and $T>1$ uses a convex-polytope-like confined space, and paSB-MSR with both $K>1$ and $T>1$ uses a union of convex-polytope-like confined spaces. 
For parSB-MSR 
with $K + T>1$, the interpretation
is the same except a convex polytope in paSB will be replaced with the complement of a convex
polytope, and vise versa.
In  contrast to SVMs using the kernel trick, MSRs using the original covariates might be more appealing in research areas, like biostatistics and sociology, where the interpretation of regression coefficients and investigation of causal relationships are of interest. In addition, 
we find that 
the classification capability of MSRs could be further enhanced with data transformation, as will be discussed in 
Section~\ref{data_trans}. 

\section{Bayesian Multinomial Support Vector Machine}\label{sec:msvm}
Support vector machines (SVMs) are max-margin binary classifiers that typically minimize 
a regularized hinge loss objective function as
\begin{align*}
l(\betav,\nu)=\sum_{i=1}^N \max(1-b_i \xv_{i}'\betav,0)+\nu R(\betav),
\end{align*}
where $b_i\in\{-1,1\}$ represents the binary label for the $i$th observation, $R(\betav)$ is a regularization function that is often set as the $L_1$ or $L_2$ norm of $\betav$, 
 $\nu$ is a tuning parameter, 
 and $\xv_{i}'$ is the $i$th row of the design matrix $\Xmat=(\xv_1,\ldots,\xv_n)'$. For linear SVMs, $\xv_{i}$ is the covariate vector of the $i$th observation, whereas for nonlinear SVMs, one typically set the $(i,j)$th element of $\Xmat$ as the kernel distance between the covariate vector of the $i$th observation and the $j$th support vector. 
The decision boundary of a binary SVM is $\{\xv : \xv'\betav =0\}$ and an observation is assigned the label $y_i=\mbox{sign}(\xv'\betav)$, which means $b_i=1$ if $\xv'\betav\ge 0$ and $b_i=-1$ if $\xv'\betav<0$.

\subsection{Bayesian Binary SVMs}
It is shown in \citet{polson2011data} 
that the exponential of the negative of the hinge loss can be expressed as a location-scale mixture of normals as 
\begin{align}
L(b_i\given \xv_i, \betav)&=\exp\left[-2\max(1-b_i\xv_i' \betav, 0)\right]\notag\\
&=\int_0^\infty \frac{1}{\sqrt{2\pi\omega_i}}\exp\left[-\frac{1}{2}\frac{(1+\omega_i-b_i\xv_i' \betav)^2}{\omega_i}\right]d\omega_i.
\notag
\end{align}
Consequently, $L(\bv\given \Xmat, \betav)=\prod_i L(b_i\given \xv_i, \betav)=\exp\left\{-2\sum_i \max(1-b_i\xv_i' \betav, 0) \right\}$ can be regarded as a pseudo likelihood in the sense that it is unnormalized with respect to $\bv =(b_1,\ldots,b_N)'\in \{-1,1\}^N$. This location-scale normal mixture representation of the hinge loss allows developing close-form Gibbs sampling update equations for the regression coefficients $\betav$ via data augmentation, as discussed in detail in \citet{polson2011data} and further generalized in \citet{henao2014bayesian} to construct nonlinear SVMs amenable to Bayesian inference. While data augmentation has made it feasible to develop Bayesian inference for SVMs, it has not addressed a common issue that SVMs provide the predictions of deterministic class labels but not class probabilities. For this reason, below we discuss how to allow SVMs to predict class probabilities while maintaining tractable Bayesian inference via data augmentation. 

Following \citet{sollich2002bayesian} and \citet{mallick2005bayesian}, by defining the joint distribution of $\betav$ and $\{\xv_i\}_i$ to be proportional to $\prod_i [L(1\given \xv_i, \betav)+L(-1\given \xv_i, \betav)]$, 
 one may define the conditional distribution of the binary label $b_i\in\{-1,1\}$ as 
\begin{align}\vspace{-3mm}\label{stick_prob}
P(b_i\given \xv_i, \betav)= 
\begin{cases}
\displaystyle \frac1 {1+e^{-2b_i\xv_i'\betav}}, & \mbox{ for } |\xv_i'\betav| \le 1;\\
\displaystyle \frac 1 {1+e^{-b_i[\xv_i'\betav+ \mathrm{sign}(\xv_i'\betav)]}}, & \mbox{ for } |\xv_i'\betav|> 1;\\
\end{cases}
\end{align}
which defines a probabilistic inference model that has the same maximum a posteriori (MAP) solution as that of a binary SVM for a given data set. Note that for MAP inference, the penalty term $\nu R(\betav)$ of the regularized hinge loss can be related to a corresponding prior distribution imposed on $\betav$, such as Gaussian, Laplace, and spike-and-slab priors 
\citep{polson2011data}. 

\subsection{paSB Multinomial Support Vector Machine}\label{pasb_svm}
Generalizing previous work in constructing Bayesian binary SVMs, we propose multinomial SVM (MSVM) under the paSB framework that is distinct from previously proposed MSVMs \citep{crammer2002algorithmic, lee2004multicategory,liu2011reinforced}.
A Bayesian MSVM that predicts class probabilities has also been proposed before in \citet{zhang2012bayesian}, which, however, does not have a data augmentation scheme to sample the regression coefficients in closed form, 
and consequently, 
relies on a random-walk Metropolis-Hastings procedure that may be difficult to tune. 

Redefining the label sample space 
from $b_i\in\{-1,1\}$ to $b_i\in\{0,1\}$, we may rewrite \eqref{stick_prob} as $b_i\given \xv_i, \betav \sim \mbox{Bernoulli}[\pi_{i,\,\mathrm{svm}}(\xv_i,\betav)]$, where 
 \begin{align}\vspace{-3mm}\label{stick_prob1}
&\pi_{i,\,\mathrm{svm}}(\xv_i,\betav)=
\begin{cases}
\displaystyle\frac1 {1+e^{-2\xv_i\betav}}, & \mbox{ for } |\xv_i'\betav| \le 1;\\
\displaystyle \frac 1 {1+e^{-\xv_i\betav- \mathrm{sign}(\xv_i'\betav)}} , & \mbox{ for } |\xv_i'\betav|> 1.\\
\end{cases} 
\end{align}
The Bernoulli likelihood based cross-entropy-loss binary classifier, whose covariate-dependent probabilities are parameterized as in \eqref{stick_prob1}, is exactly what we need to extend the binary SVM into a multinomial classifier under paSB introduced in 
{Theorem \ref{thm1}}.
More specifically, given the category-stick mapping $\zv$, with the success probabilities of the stick that category $s$ is mapped to parameterized as $\pi_{iz_s,\,\mathrm{svm}}(\xv_i,\betav_s)$ and binary stick variables drawn as $b_{iz_s}\sim\mbox{Bernoulli}[\pi_{iz_s,\,\mathrm{svm}}(\xv_i,\betav_s)]$, we have the following definition. 

\begin{definition}[paSB multinomial SVM] Under the paSB construction, given the covariate vector $\xv_i$ and category-stick mapping $\zv$, multinomial support vector machine (MSVM) parameterizes $p_{is}$, the multinomial probability of category $s$, as
\begin{align}\vspace{-3mm} \notag 
p_{is}(\zv) 
=[\pi_{iz_s,\,\mathrm{svm}}(\xv_i,\betav_s)]^{\mathbf{1}(z_s\neq S)}\prod \nolimits_{j:z_j<z_s}\pi_{iz_j,\,\mathrm{svm}}(\xv_i,\betav_j).
\end{align} 
\end{definition}
Note that there is no need to introduce parSB-MSVM in addition to paSB-MSVM, since by definition, we have $\pi_{iz_s,\,\mathrm{svm}}(\xv_i,-\betav_s) = 1- \pi_{iz_s,\,\mathrm{svm}}(\xv_i,\betav_s)$ for all $s$.


\section{Example Results}\label{eg}

Constructed under the paSB framework, a multinomial regression model of $S$ categories is characterized by not only how the $S$ stick-specific binary classifiers with cross entropy loss parameterize their covariate-dependent probability parameters, but also how its $S$ categories are one-to-one mapped to 
$S$ latent sticks. To investigate the unique 
properties of a paSB multinomial regression model, we will study the benefits of both inferring an appropriate 
mapping $\zv$ 
and increasing 
 the modeling capacity 
 of the underlying binary regression model. 
For illustration purpose, we will focus on multinomial softplus regression (MSR) whose capacity and complexity are both explicitly controlled by $K$ and $T$.


\subsection{
Influence of Binary Regression Model Capacity
}\label{nonlinear_msr}

We first consider the Iris data set with $S=3$ categories. We choose the sepal and petal lengths as the two dimensional covariates to illustrate the performance of MSR under four different settings. We fix $\zv=(1,2,3)$, which means category $s$ is mapped to stick $s$ for all~$s$, but choose different model capacities by varying $K$ and $T$. 

Examining the relative 2D spatial locations of the observations, where the blue, black, and gray points are labeled as category 1, 2, and 3, respectively, one can imagine that setting $\zv=(2,1,3)$, which means mappings categories 2, 1, and 3 to the 1st, 2nd, and 3rd sticks, respectively, will already lead to excellent class separations for MSR with $K=T=1$, according to the analysis in Section \ref{sec:geometric} and also confirmed by our experimental results (not shown for brevity). 
More specifically, 
with the 2nd, 1st, and 3rd categories mapped to the 1st, 2nd, and 3rd sticks, respectively, one can first use a single hyperplane to separate category 2 (black points) from both categories 1 (blue points) and 3 (gray points), and then use another hyperplane to separate category 1 (blue points) from category 3 (gray points). 

However, when the mapping is fixed at $\zv=(1,2,3)$, as shown in the first row of Figure~\ref{iris}, MSR with $K=T=1$ performs poorly and fails to separate out category 1 (blue points) right in the beginning. This is not surprising since MSR with $K=T=1$ is only equipped with a single hyperplane to separate the category that the first stick is mapped to (category $z_1=1$ in this case) from the others, whereas for this data set it is apparent at least two hyperplanes are required to separate the blue from the black and gray points. 
MSR with $K=5$ and $T=1$ also fails to work with $\zv=(1,2,3)$, as shown in the third row of Figure \ref{iris}, which is also not surprising since it can only use the complementary of a convex-polytope-bound confined space to enclose category $z_1=1$, but the blue points can not be enclosed in such a manner.
Despite purposely enforcing an unfavorable category-stick mapping, 
once we increase $T$, 
the performance quickly improves, which is expected since $T>1$ allows using a single (if $K=1$ as in the second row) or a union (if $K>1$ as in the fourth row) of convex-polytope-like confined spaces to separate one category from the others (by enclosing the positively labeled observations in each stick-specific binary classification task). 

\begin{figure}[!t]
\centering
 \centering
 \includegraphics[width=.7\columnwidth]{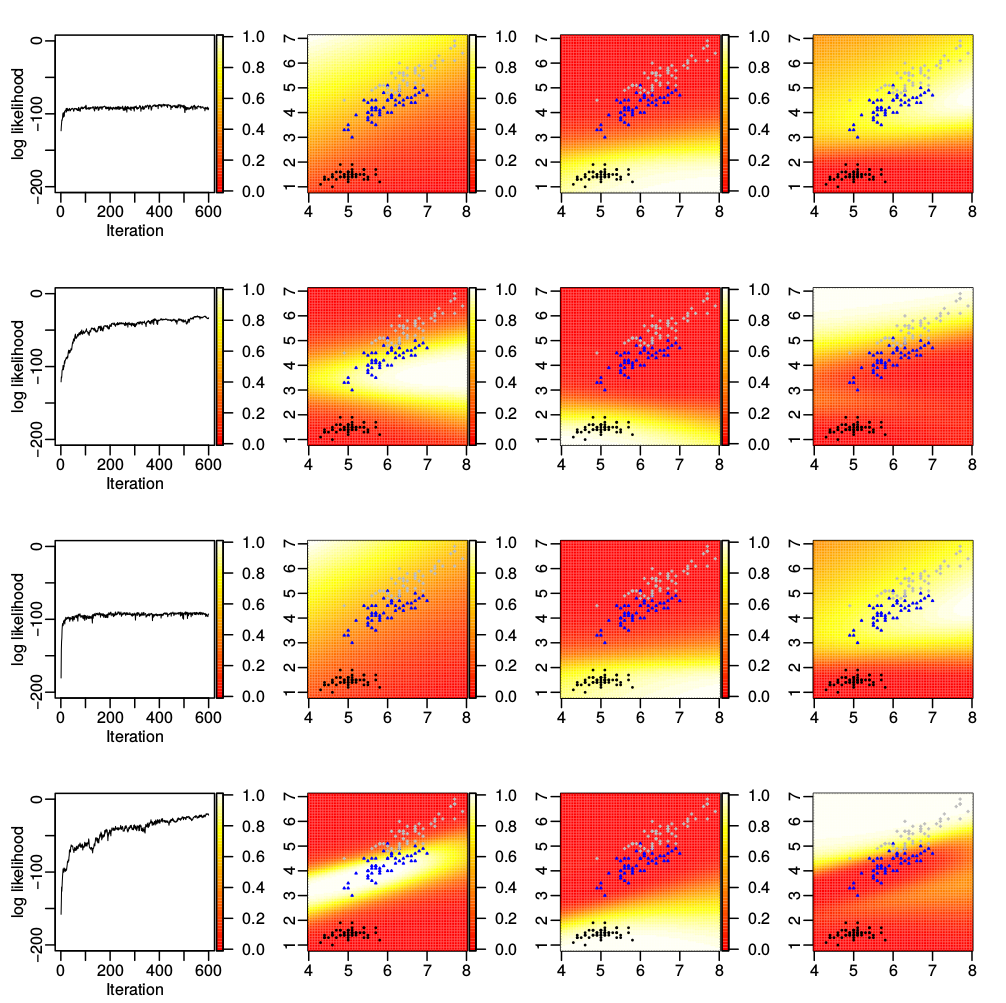}
 \caption{Log-likelihood plots and predictive probability heat maps for the 2-D iris data with a fixed category-stick mapping $\zv=(1,2,3)$. Blue, black, and gray points are labeled as categories 1, 2, and 3, respectively. For the first row, $K=1$ and $T=1$, second row, $K=1$ and $T=3$, third  row, $K=5$ and $T=1$, and fourth row, $K=5$ and $T=3$. The log-likelihood plots are shown in Column 1,
and the predictive probability heat maps of categories $1$ (blue), $2$ (black), and $3$ (gray) are shown in  Columns 2, 3, and 4, respectively. }\vspace{-3mm}\label{iris}
 \end{figure}

The results in Figure \ref{iris} show that even an unoptimized category-stick mapping, which is unfavorable to MSR with small $K$ and/or $T$, is enforced, empowering each stick-specific binary regression model with a higher capacity (using larger $K$ and/or $T$) can still allow MSR to achieve excellent separations. 
It is also simple to show that for the data set in Figure \ref{iris}, even if one chooses low-capacity stick-specific binary regression models by setting $T=1$, one can still achieve good performance with MSR if the category-stick mapping is set as $\zv=(2,1,3)$, $\zv = (3,1,2)$, $\zv=(2,3,1)$, or $\zv = (3,2,1)$. That is to say, as long as it is not category 1 (blue points) that is mapped to stick 1, MSR with $T=1$ is able to provide satisfactory performance. 

\subsection{Influence of Category-Stick Mapping and its Inference}\label{permute_z}

\begin{figure}[!t]
 \centering
 \includegraphics[width=.7\columnwidth]{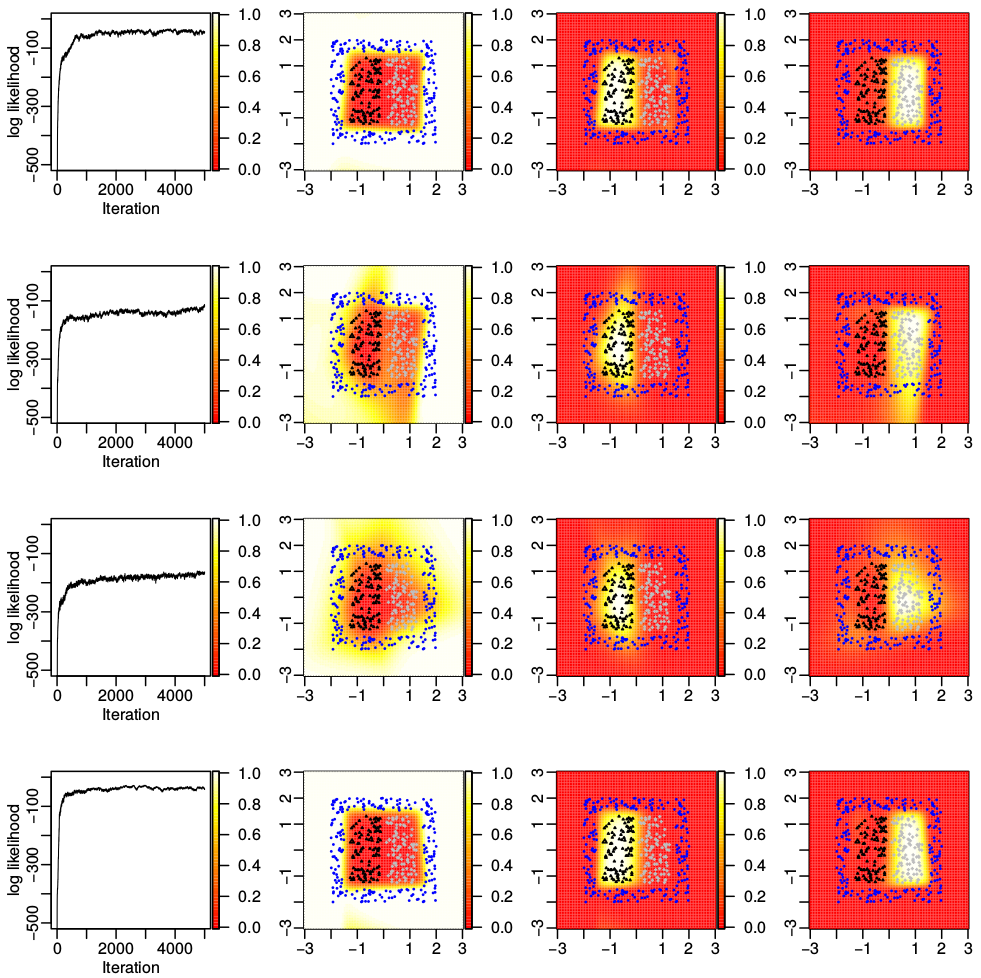}
 \caption{ Log-likelihood plots and predictive probability heat maps for the square data with $K=T=10$. The blue, black, and gray points are labeled as categories 1, 2, and 3, respectively. We fix the category-stick mapping as $\zv=(1,2,3)$ for Row 1, $(2,1,3)$ for Row 2, and $(3,1,2)$ for Row 3, and sample $\zv$ for Row 4. The log-likelihood plots are shown in Column 1, and   the predictive probability heat maps of categories 1 (blue), 2 (black), and 3 (gray) are shown in Columns 2, 3, and 4, respectively.}
 \label{square}
\end{figure}

The Iris data set in Figure \ref{iris} provides an instructive  example to show not only the importance of increasing the model capacity if a poor category-stick mapping is imposed, but also the importance of optimizing the category-stick mapping if the capacities of these stick-specific binary regression models are limited. To further illustrate the benefits of inferring an appropriate category-stick mapping $
\zv$, we consider the \textit{square} data set shown in Figure \ref{square}. {We show that for MSR, 
even if both $K$ and $T$ are sufficiently large to allow each stick-specific binary regression model to have a high enough capacity, whether an optimal category-stick mapping is selected may still clearly matter for the performance.} 

As shown in the first three rows of Figure \ref{square}, with $K=T=10$, three different $\zv$'s are considered and $\zv=(1,2,3)$ (shown in the first row) is found to perform the best. As shown in the fourth row, we sample $\zv$ using \eqref{eq:MH} within each MCMC iteration and 
achieve a result that seems as good as fixing $\zv=(1,2,3)$. In fact, we find that our inferred mappings switch between $\zv=(1,2,3)$ and $\zv=(1,3,2)$ during MCMC iterations, indicating that the Markov chain is mixing well. These results suggest the importance of both learning the mapping $\zv$ from the data and allowing the stick-specific binary classifiers to have enough capacities to model nonlinear classification decision boundaries.

When sampling $\zv=(z_1,\ldots, z_S)$ that the $S$ categories are mapped to, 
although $S!$ permutations of $(1,\ldots,S)$ can become enormous as $S$ increases, the effective search space could be much smaller if many different mappings imply similar likelihoods and if these extremely poor mappings can be easily avoided. Rather than searching for the best mapping, the proposed MH step, proposing two indices $z_j$ and $z_{j'}$ to switch in each iteration, is a simple but effective strategy to escape from the mappings that lead to poor fits. Note that the probability of a $z_j$ not being proposed to switch after $t$ MCMC iterations is $[(S-2)/S]^t$. Even if $S$ is as large as 100, this probability is less than $10^{-8}$ at $t=1000$.  Also note the iteration at which $z_j$ is proposed to switch at the first time follows a geometric distribution, with success probability 
$2/S$. Thus $S/2$ is the expected number of iterations for a  $z_j$ to be proposed to switch once.

To demonstrate the efficiency of our permutation scheme, we construct square101, a synthetic two-dimensional data set   consisting of 101 categories. We generate 8000 data points that are uniformly at random distributed within the $12\times 12$ spatial region occupied by all 101 categories.  The decision boundaries of different classes are displayed in Figure \ref{fac_100}(a), where the data points placed within the outside square frame, whose outer and inner dimensions are $12$ and $10$, respectively, are assigned to category 1, and these  placed 
within the $s$th unit square, where $s\in\{2,\ldots,101\}$, inside the square frame are assigned to category~$s$.  Although it is almost impossible to search for the best category-stick mapping $\zv$ giving rise to the highest likelihood from all $101!\approx 10^{160}$ possible mappings, we show our permutation scheme is  very effective in escaping from poor mappings, leading to a performance that is comparable to the best of those obtained with 
pre-fixed suboptimal mappings. 
More specifically, 
applying the analysis in Section \ref{sec:geometric} to Figure \ref{fac_100}(a), we expect an aSB-MSR to perform well under a fixed suboptimal category-stick mapping $\zv$, where  $z_1=1$, which means 
 the outside square frame is mapped to stick 1, and the 
squares 
closer to the inner boundary of the square frame are mapped to the sticks broken at earlier stages; the mapping $\zv=(1,2,\cdots,101)$ is such an example. In other words, we first separate the frame from all the other squares, and then sequentially separate the squares from the remainders; the closer a square is from the frame, the earlier it is separated. The total number of suboptimal mappings $\zv$'s constructed in this manner is as large as $36!\times 28! \times 20! \times 12! \times 4! \approx 10^{99.5} $.

First, we uniformly at random generate 3600 different suboptimal mappings $\zv$'s under this construction, run aSB-MSR with $K=T=4$, 
and plot the histogram of the 3600 log-likelihoods in Figure \ref{fac_100}(b). Second, we start from 3600 randomly initialized $\zv$, 
run paSB-MSR with $K=T=4$, and also plot the histogram of the 3600 log-likelihoods in Figure \ref{fac_100}(b). 
 For each run, we choose 20,000 MCMC iterations and collect the last 1000 MCMC samples. Each log-likelihood is averaged over those of the corresponding model's collected MCMC samples. 
As in Figure \ref{fac_100}(b), the log-likelihood from a paSB-MSR is in general clearly larger than that of an aSB-MSR with a fixed suboptimal $\zv$, and there is little overlap between their corresponding histograms.  
Further examining the 3600 $\zv$'s inferred by paSB at its last MCMC iteration shows that
3482 of them have $z_1=1$ and all of them have $z_1\leq 5$. 
Suppose $z_1\notin\{1,2,3,4,5\}$ at the current iteration, which means category 1 is mapped to none of the first five sticks, then the probability of not only selecting stick $z_1$, but also switching it with one of the first five sticks in the MH proposal is $\frac 1 {101}\times \frac{5}{100}$. Thus the probability that category 1 
 has never been proposed to mapped to one of the first five sticks 
  after $t$ iterations is $\left[1-{5}/({101\times 100})\right]^t$,  which becomes as small as $0.005\%$ at $t=20,000$, demonstrating the effectiveness  of our permutation scheme in dealing with a large number of categories. 
Note we have also tried 3600 aSB-MSR, each of which is provided with a randomly initialized $\zv$. The log-likelihoods, however, are all far below $-4000$ and hence not included for comparison. This phenomenon is not surprising, as the probability for a randomly initialized $\zv$ to be suboptimal is as tiny as $36!\times 28! \times 20! \times 12! \times 4! /101!\approx 10^{-60.5}$.

\begin{figure}[t!]
 \centering
 \begin{subfigure}[t]{0.4\textwidth}
 \centering
 \includegraphics[width=1\linewidth]{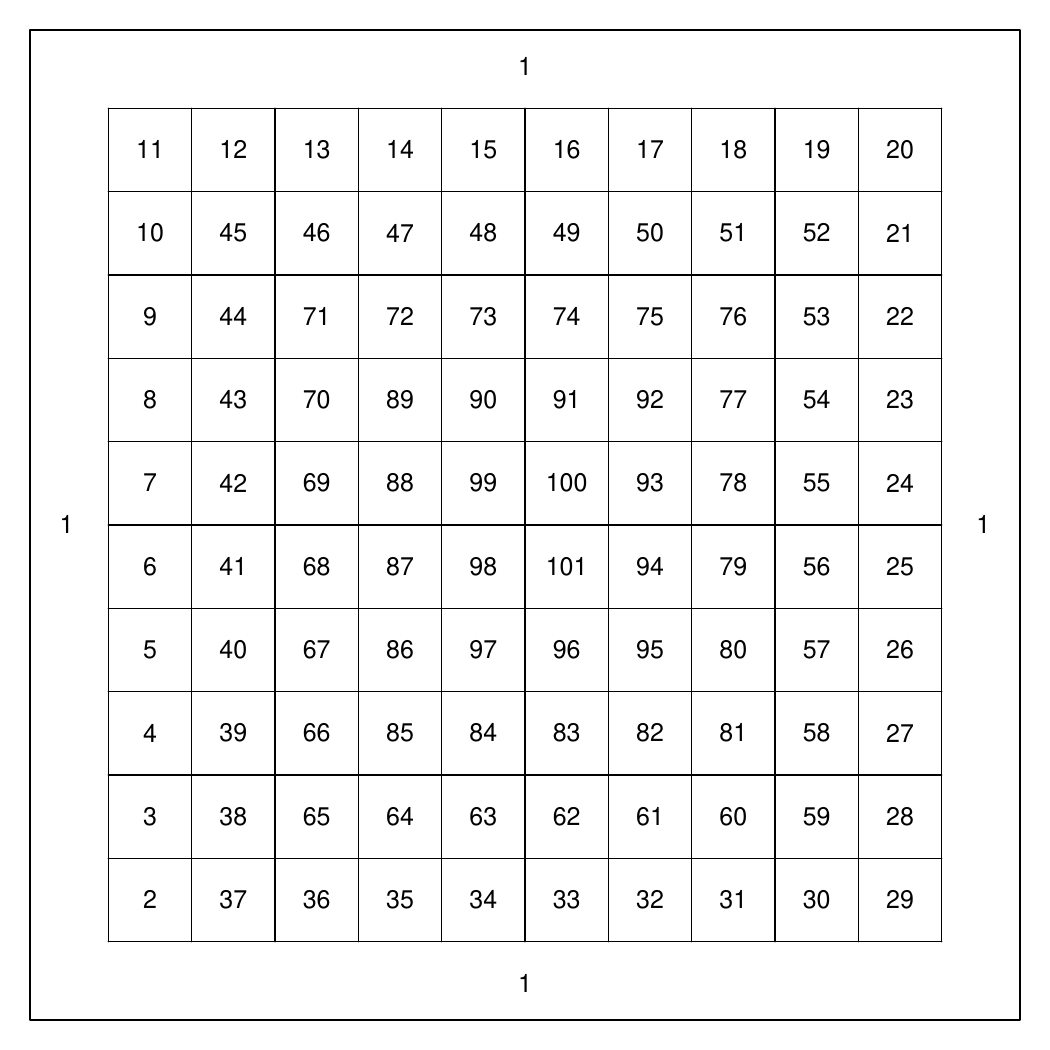}\vspace{-3mm}\label{fac_100_data}
 \caption{ }\vspace{-2mm}
 \end{subfigure}%
 ~ 
 \begin{subfigure}[t]{0.4\textwidth}
 \centering
 \includegraphics[width=1\linewidth]{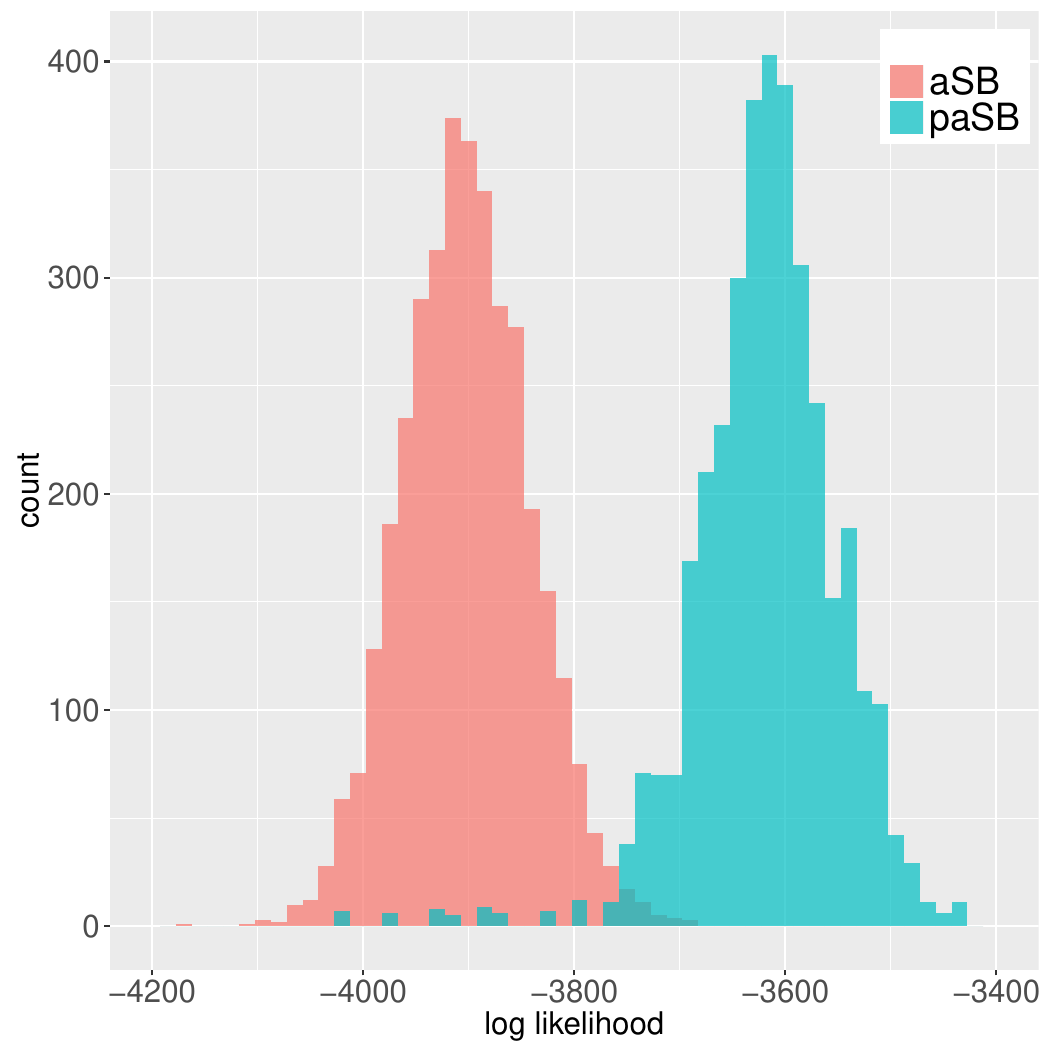}\vspace{-3mm}\label{fac_100_hist}
 \caption{ }\vspace{-2mm}
 \end{subfigure}%
 \caption{(a) Illustration of the square101 data and (b) log-likelihood histograms, by aSB-MSR with 3600 random 
 suboptimal category-stick mappings and by paSB-MSR with 3600 randomly initialized category-stick mappings. 
 }\label{fac_100}
\end{figure}

Figure \ref{satimge_hist} empirically demonstrates the effectiveness of permuting $\zv$ on the satimage data set, using MSRs with $K=5$, $T=3$, and $\zv$ fixed at each of the $6!=720$ possible one-to-one category-stick mappings. 
Panels (a) and (b) show the log-likelihood histograms 
 for MSRs constructed under augmented SB (aSB) and augmented reversed SB (arSB), respectively. Both histograms are clearly left skewed, indicating under both aSB and arSB, only a small proportion of the 720 different category-stick mappings lead to very poor fits. The blue vertical lines at $-1203.82$ in (a) and $-1350.21$ in (b) are 
 the log-likelihoods by paSB and parSB, respectively, in both of which the category-stick mapping $\zv$ is updated by a MH step in each MCMC iteration. 
Only 20 (97) out of 720 aSB-MSRs (arSB-MSRs) have a higher likelihood than paSB-MSR (parSB-MSR). 
\begin{figure}[t!]
 \centering
 \begin{subfigure}[t]{0.4\textwidth}
 \centering
 \includegraphics[width=1\linewidth]{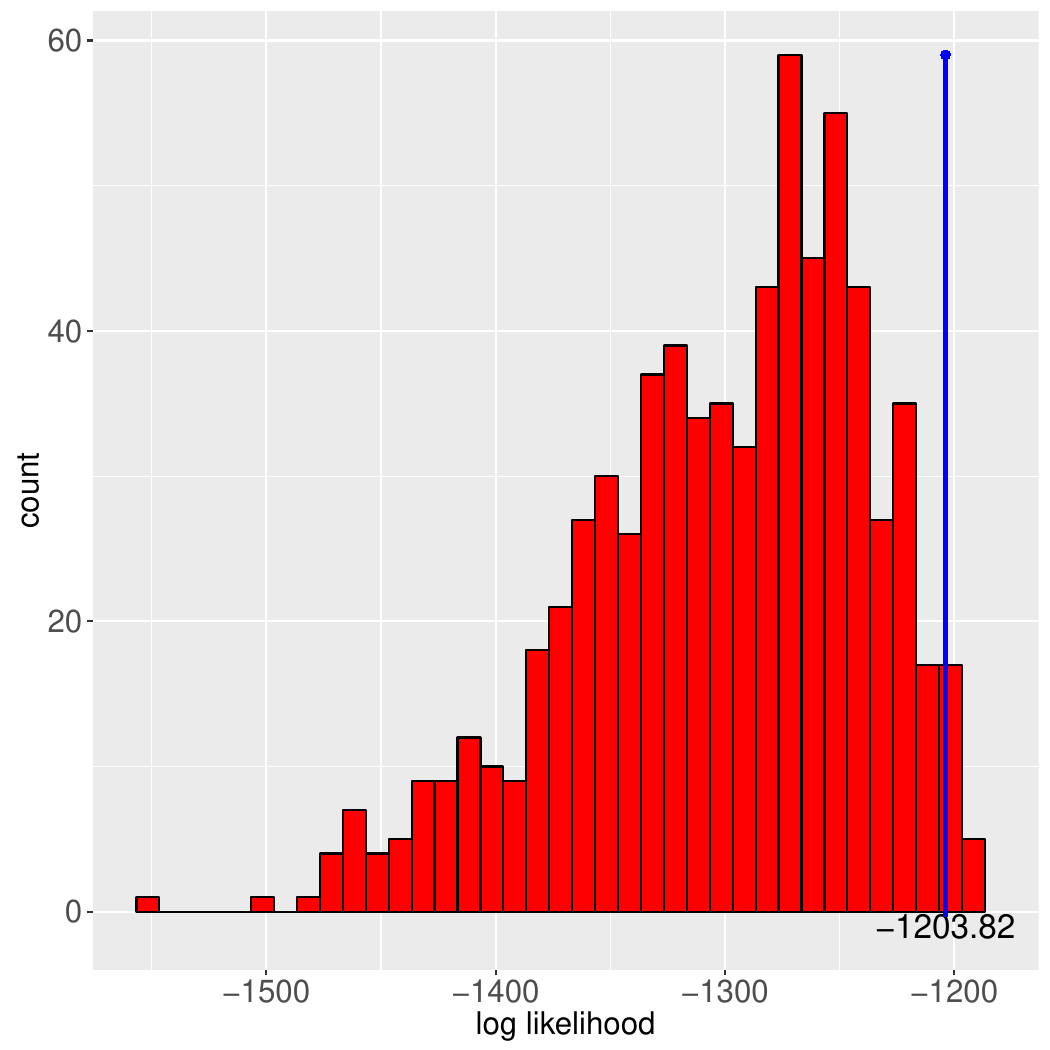}\vspace{-3mm}\label{720pasb}
 \caption{aSB}
 \end{subfigure}%
 ~ 
 \begin{subfigure}[t]{0.4\textwidth}
 \centering
 \includegraphics[width=1\linewidth]{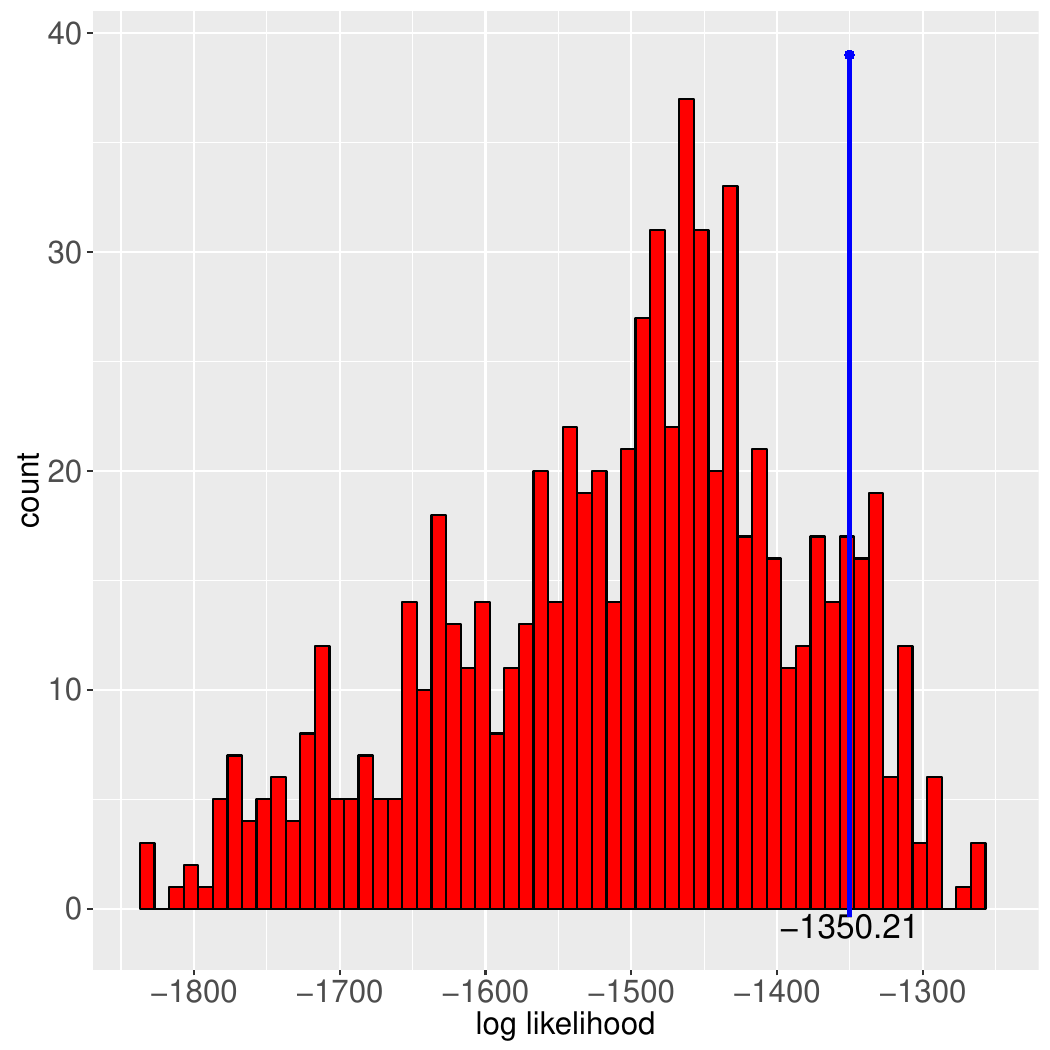}\vspace{-3mm}\label{720parsb}
 \caption{arSB}
 \end{subfigure}%
 \caption{Log-likelihood histograms for MSRs using all 720 possible category-stick mappings, constructed under (a) augmented stick breaking (aSB) and (b) augmented and reversed stick breaking (arSB). The blue lines in (a) and (b) correspond to the log-likelihoods of paSB-MSR and parSB-MSR, respectively. 
 }\label{satimge_hist}
\end{figure}

Since in the stick-breaking construction, the binary classifier that separates a category mapped to a smaller-indexed stick from the others 
utilizes fewer constraints, the classification can be poor if the complexity of the decision boundary goes beyond the nonlinear modeling capacity of the binary classifier. However, even with a low-capacity binary classifier, the performance could be significantly improved if that difficult-to-separate category is mapped to a larger-indexed stick, for which there are fewer categories left to be separated in its ``one-vs-remaining'' binary classification problem. 
Examining the $\zv$'s associated with the 100 lowest log-likelihoods in Figure \ref{satimge_hist}, we find there are 51 mappings belonging to the set $\{ \zv: z_5=1$ or $z_6=1 \}$ in aSB, and 77 belonging to $ \{ \zv: z_3=1$ or $z_6=1\}$ in arSB. It suggests that separating Categories 5 or 6 (Categories 3 or 6) from all the other categories might be beyond the capacity of a binary softplus regression with $K=5$ and $T=3$ under the aSB (arSB) construction. But if breaking the sticks associated with these categories at late stages, we only need to separate them from fewer remaining categories, which could be much easier. We have further examined the other 620 arrangements, and found no evident patterns. 
These observations suggest that the effective search space of the mapping $\zv$ is considerably smaller than $S!$, and the proposed MH step is effective in escaping from poor category-stick mappings.

In paSB-MSVM, we use a Gaussian radial basis function kernel, whose kernel width is cross validated from a set of predefined candidates. We find its performance to be sensitive to the setting of the kernel width, which is a common issue for SVMs \citep{cherkassky2004practical, soares2004meta,chang2005scaling}. If an appropriate kernel width could be identified through cross validation, we find that learning the mapping $\zv$ becomes less important for paSB-MSVM to perform well. However, we find that if the kernel width is not well selected, which can happen if all candidate kernel widths are far from the optimal value, the binary classifier for each category may not have enough capacity for nonlinear classification and the learning of the category-stick mapping $\zv$ could then become 
 important.

\subsection{Turning Off Unneeded Model Capacities}

 \begin{figure}[t]
\centering
 \includegraphics[width=.44\columnwidth]{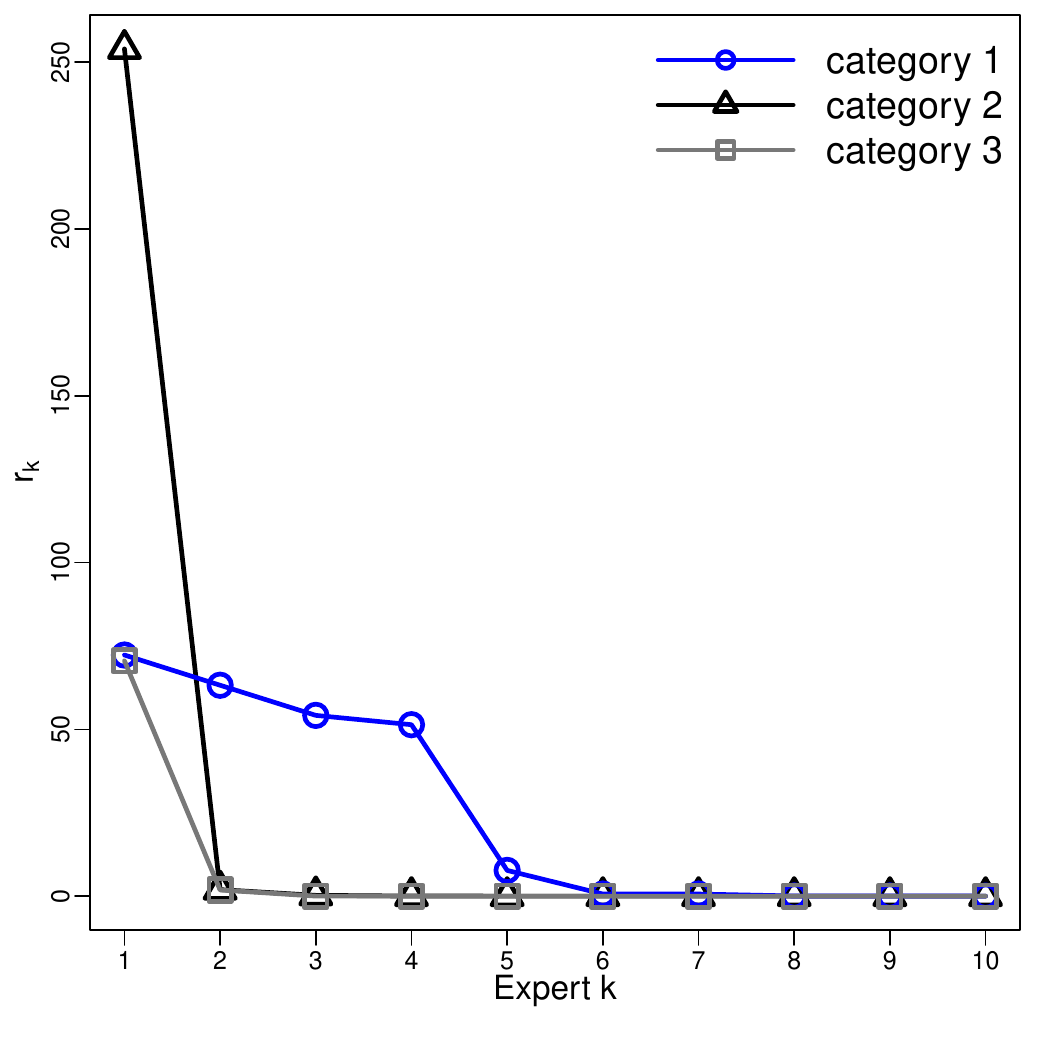}
 \caption{Inferred expert weights $r_k$ in descending order for each category of the square data with $K=T=10$.} 
 \label{rk}
\end{figure}

While one can adjust both $K$ and $T$ to control the capacity of binary softplus regression, for MSR, 
 the total number of experts $K$ is a truncation level that can be set as large as permitted by the computation budget. This is because the truncated gamma process used by each stick-specific binary softplus regression  shrinks 
 the weights of unnecessary experts 
 towards zeros. 
 Figure \ref{rk} shows in decreasing order the inferred weights of the experts belonging to each of the 3 categories of the square data set. These weights are inferred by MSR with $K=T=10$ and the learning of $\zv$, 
 as in the fourth row of Figure \ref{square}. It is clear from Figure \ref{rk} that only a small number of experts are inferred with non-negligible weights in the posterior, and the number of active experts and their weights indicate the complexity of the corresponding classification decision boundaries shown in the fourth row of Figure \ref{square}. 
We note that while $T$ is a parameter to be set by the user, we find increasing it increases model capacity, without observing clear signs of overfitting for all the data considered here. 

  %

\subsection{MSR with Data Transformation}\label{data_trans}

\begin{figure}[t]
\centering
\includegraphics[width=1\columnwidth]{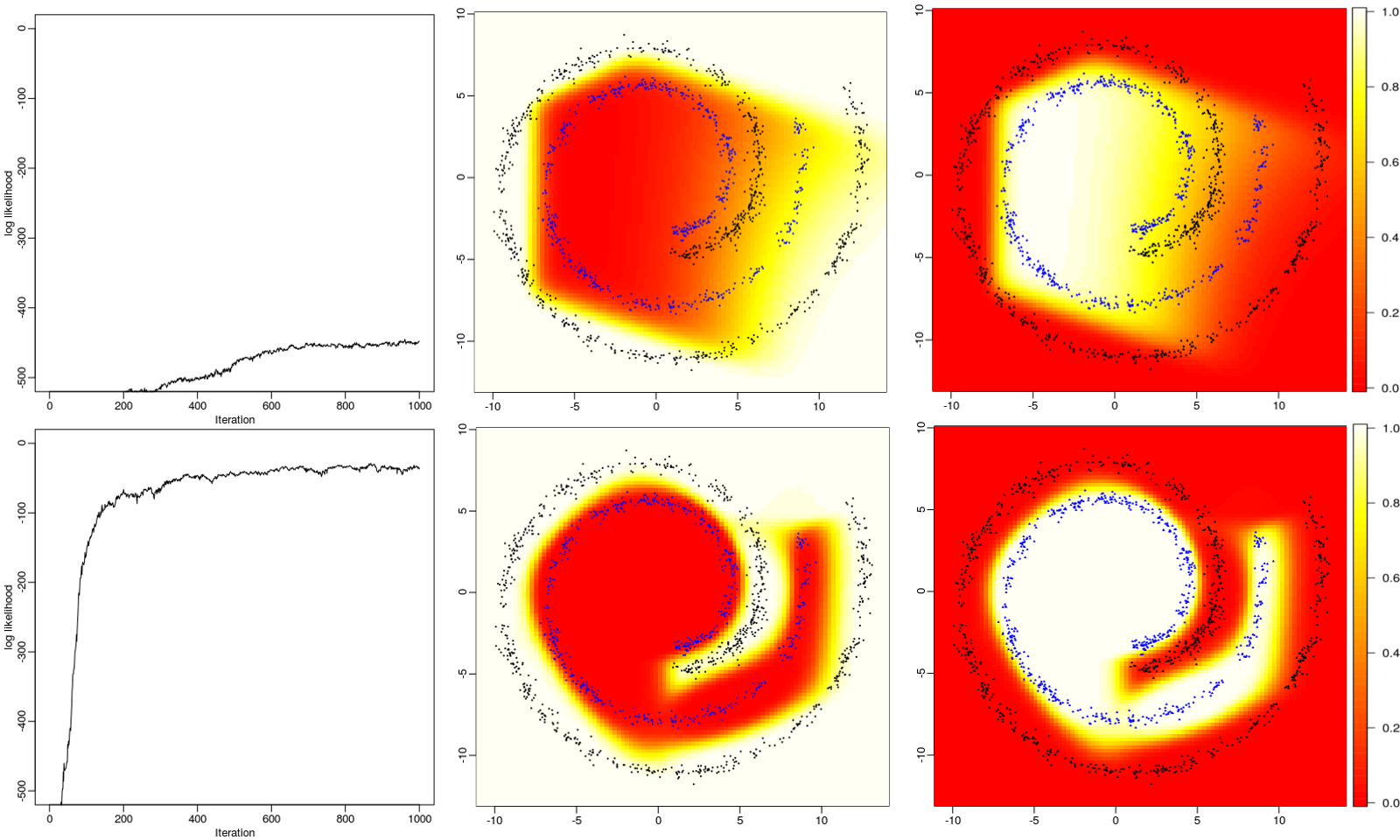}
\caption{First row: classification of a 2-D swiss roll data by paSB-MSR with $K=5$, $T=3$, using the original covariates. Second row: paSB-MSR with $K=5$, $T=1$ trained on the covariates transformed via the paSB-MSR  used in the first row.  
In each row, the left column plot the log-likelihood against MCMC iteration, and the middle and right columns  show the predictive probability heatmaps for Category 1 (black points) and Category 2 (blue points), respectively.} 
\label{swissroll}
\end{figure}

Kernel 
SVMs transform the data to make different categories more linearly separable in the transformed covariate space. 
While kernel SVMs may provide high nonlinear modeling capacity, its performance could be sensitive to the kernel width, which often needs to be cross validated, and its number of support vectors often increases linearly in the size of the training set. 
By contrast, MSRs rely on the interactions of linear hyperplanes to construct nonlinear decision boundaries,  as discussed in Section \ref{sec:geometric}, and hence may have insufficient capacity for highly complex nonlinearity. However, we may simply stack another MSR on a previously trained MSR to quickly enhance its nonlinear modeling capacity. In particular, we may first run a MSR to obtain a finite set of hyperplanes denoted by $\tilde{\betav}_{jk}^{(t+1)}$. We may then augment the original covariate vector $\xv_i$ as
\begin{align}
\tilde{\xv}_i:=\left[ \xv_i', \log\big(1+e^{\xv_i'  \tilde{\betav}_{11}^{(2)}}\big),
 \cdots, \log\big(1+e^{\xv_i'  \tilde{\betav}_{jk}^{(t+1)}}\big),\cdots, 
 \log\big(1+e^{\xv_i' \tilde{\betav}_{SK}^{(T+1)}}\big)\right]'\label{eq:data_trans}
\end{align}
and run another MSR with the  transformed covariates $\tilde{\xv}_i$. 

%

For illustration, we show the efficacy of this data-transformation strategy on a 2-D swiss roll data in Figure \ref{swissroll}. The first row shows the results of MSR with $K=5$ and $T=3$, using the original covariates $\xv_i$, while the second row shows MSR with $K=5$ and $T=1$, using the transformed covariates $\tilde{\xv}_i$ defined by \eqref{eq:data_trans}, where the regression coefficient vectors $\tilde{\betav}_{jk}^{(t+1)}$ are learned using the MSR illustrated in the first row. It is evident that  the classification is greatly improved in terms of both training log-likelihood and out-of-sample predictions. 

\subsection{Results on Benchmark Data Sets}\label{comparison}

To further evaluate the performance of the proposed paSB multinomial regression models, we consider paSB multinomial logistic regression (paSB-MLR), paSB multinomial robit with $\kappa=6$ degrees of freedom  (paSB-robit), paSB multinomial support vector machine (paSB-MSVM), and MSRs. We compare their performance with those of $L_2$ regularized multinomial logistic regression ($L_2$-MLR), support vector machine (SVM), and adaptive multi-hyperplane machine (AMM), 
and consider the following benchmark multi-class classification data sets: iris, wine, glass, vehicle, waveform, segment, 
dna, and satimage. We also include the synthetic square data shown in Figure 2 for comparison.  For SVM we use the
LIBSVM package, \color{black} which 
trains $S(S-1)/2$ one-vs-one binary classifiers and makes prediction using majority voting  \citep{LIBSVM}. We run LIBSVM in \texttt{R} with  package \texttt{e1071} \citep{e1071}. 
We consider MSRs with $(K,T)$ as $(1,1)$, $(1,3)$, $(5,1)$, and $(5,3)$,  respectively. We also consider MSR with data transformation (DT-MSR), in which we first train a MSR with $K=5$ and $T=3$ to transform the covariates and then stack another MSR with $K=5$ and $T=1$. We provide detailed descriptions on the data and experimental settings in the Appendix.

With the number of categories in parentheses right after the data set names, we summarize in Table \ref{compare} the classification error rates by various models, 
where those of MSRs are calculated by averaging over paSB and parSB. 
Table \ref{compare} shows that an MSR 
with $K$ or $T$ sufficiently large generally outperforms paSB-MLR,  paSB-robit, $L_2$-MLR, and AMM, 
and using another MSR on the transformed covariates 
can in general further reduce the error rate. This is especially evident when there are nonlinearly separable categories, as indicated by a clearly higher error rate of $L_2$-MLR in contrast to that of SVM. One may notice that paSB-robit, paSB-MLR, and MSR with $K=T=1$ are similar to $L_2$-MLR in terms of performance, 
suggesting the effectiveness of the proposed permutation scheme, 
which helps mitigate the potential adverse effects of having asymmetric class labels. 
One may also note that paSB-robit outperforms paSB-MLR on glass, vehicle, waveform, dna, and satimage, 
indicating there are  benefits 
in using a robust classifier on these data sets. 
Comparable error rates of paSB-MSVM to SVM and  better performance of MSRs on most data sets demonstrate the success of the paSB framework in transforming a binary classifier with cross entropy loss into a Bayesian multinomial one. 

\begin{table}[t]
\centering
\label{compare}
\makebox[\linewidth]{
\resizebox{\linewidth}{!}{%
\begin{tabular}{l
P{1.10cm}P{1.10cm}P{1.10cm}P{1.10cm}P{1.10cm}P{1.10cm}P{1.10cm}P{1.10cm}P{1.10cm}P{1.15cm}P{1.15cm}}
 \hline
Data (S) & paSB-MLR & paSB-robit& paSB-MSVM & $K=1$ $T=1$ & $K=1$ $T=3$ & $K=5$ $T=1$ & $K=5$ $T=3$ & DT-MSR & $L_2$-MLR & SVM & AMM \\ 
 \hline
square(3)& 59.52 &67.46  & \textbf{0} & 57.14 & 15.08 & \bf0 & \bf0 &\bf0 & 62.29 & 4.76 & 16.67 \\ 
 iris(3) & 4.00&5.33  &  \bf3.33 & 4.67 & 4.00 & 4.00 & 4.00 & \bf3.33 &  \bf3.33 & 4.00 & 4.67 \\ 
 wine(3) & 4.44&5.00  & 2.78 & 2.78 &  \bf2.22 & 2.78 &  \bf2.22 &2.78& 3.89 & 2.78 & 3.89 \\ 
 glass(6) & 35.35&34.88  & 29.30 & 33.49 & 26.05 & 31.16 & 32.09 & \bf 26.37& 33.02 & 28.84 & 37.67 \\ 
 vehicle(4) & 23.23&21.25  & 17.32 & 22.44 & 17.32 & 17.72 & 15.75 & \bf 14.96& 22.83 & 18.50 & 21.89 \\ 
 waveform(4) & 17.87& 16.42  & 15.76 & 19.84 & 16.62 & 15.67 &  \bf 15.04 & 15.56 & 15.60 & 15.22 & 18.54 \\ 
 segment(7)& 7.36&8.03  & 7.98 & 6.20 & 6.49 & 6.45 &  \bf 5.63 & 7.65 & 8.56 & 6.20 & 12.47 \\ 
 dna(3) & 5.06 &4.05& 5.31 & 4.13 & 4.47 & 4.55 & 4.22 & \bf 3.88 & 5.98 & 4.97 & 5.43 \\ 
 satimage(6) & 20.65&17.25 & 8.90 & 16.65 & 14.45 & 12.85 & 12.00 & 9.85 & 17.80 &  \bf 8.50 & 15.31 \\ 
 \hline
\end{tabular}
}
}
\caption{Comparison of the classification error rates (\%) of paSB-MLR, paSB-robit, paSB-MSVM, MSRs with various $K$ and $T$ (columns 5 to 8),  MSR with data transformation (DT-MSR), $L_2$-MLR, SVM, and AMM.}\label{compare}
\end{table}


To further check whether a paSB model is attractive when fast out-of-sample prediction is desired, we consider 
using only the MCMC sample that has the highest training likelihood among the collected ones for all paSB models, and 
summarize in Table \ref{compare1} of the Appendix the classification error rates of various models, with the number of inferred support vectors or active hyperplanes included in parenthesis. 
Following the definition of active experts 
in \citet{SoftplusReg_2016}, we define for MSRs the number of active hyperplanes as $T\sum_s^S \widetilde K_s$ where $\widetilde K_s$ is the number of active experts for class $s$. The number of active hyperplanes determines the computational complexity for out-of-sample prediction with a single MCMC sample, which is $O(T\sum_s^S\widetilde K_s)$. Since the error rates of MSRs in Table \ref{compare1} are calculated by averaging over both paSB and parSB, the number of active hyperplanes is $T\sum_{s}^S(\widetilde K_s^{(paSB)}+\widetilde K_s^{(parSB)})$. 

Shown in Figure \ref{boxplot} in the Appendix are boxplots of the number of each category's active experts for MSR with $K=5$ and $T=3$. 
Except for several categories of satimage that require all $K=5$ experts for parSB-MSR, $K=5$ is large enough to provide the 
needed model capacity under all the other scenarios. 
As shown in Table \ref{compare1}, MSRs with sufficiently large $K$ and/or $T$ are comparable to both SVM and paSB-MSVM in terms of the error rates, while clearly outperforming them in terms of the number of (active) hyperplanes/support vectors and hence computational complexity for out-of-sample predictions. While MSR with $K=T=1$,  paSB-MLR, and paSB-robit generally perform worse than SVM in terms of the error rates, they use much fewer hyperplanes and hence have significantly lower computation for out-of-sample predictions. In summary, MSR whose upper-bound for the number of active expects $K$ and number of layers for each expert $T$ can both be adjusted to control its capacity of modeling nonlinearity, can achieve a good compromise between the accuracy and computational complexity for out-of-sample prediction of multinomial class probabilities, and can be further improved by training an additional MSR on the transformed covariates.

\begin{table}[t]
\centering
\makebox[\linewidth]{
\resizebox{\linewidth}{!}{
\begin{tabular}{lP{1.10cm}P{1.10cm}P{1.10cm}P{1.10cm}P{1.10cm}P{1.10cm}P{1.10cm}P{1.10cm}P{1.10cm}P{1.10cm}P{1.10cm}P{1.10cm}}
  \hline
  &  \multicolumn{4}{c}{$10\%$ quantile}  &  \multicolumn{4}{c}{median} &  \multicolumn{4}{c}{$90\%$ quantile} \\
&  \multicolumn{2}{c}{training}  &  \multicolumn{2}{c}{testing} &  \multicolumn{2}{c}{training}  &  \multicolumn{2}{c}{testing} &  \multicolumn{2}{c}{training}  &  \multicolumn{2}{c}{testing} \\  
& Bayes MLR & paSB-MLR & Bayes MLR & paSB-MLR & Bayes MLR & paSB-MLR & Bayes MLR & paSB-MLR & Bayes MLR & paSB-MLR & Bayes MLR & paSB-MLR \\ 
  \hline
square & 411.46 & 194.54 & 421.62 & 256.48 & 913.17 & 948.53 & 858.46 & 952.33 & 924.48 & 973.60 & 927.93 & 976.33 \\ 
  iris & 82.30 & 90.33 & 73.75 & 84.40 & 149.47 & 218.21 & 156.25 & 174.16 & 331.71 & 854.45 & 341.39 & 793.00 \\ 
  wine & 194.41 & 314.41 & 58.41 & 56.30 & 467.73 & 859.67 & 506.66 & 643.90 & 926.55 & 991.46 & 958.12 & 994.77 \\ 
  glass & 70.18 & 162.69 & 67.81 & 138.90 & 137.18 & 335.14 & 122.27 & 329.21 & 359.75 & 686.43 & 347.40 & 615.02 \\ 
  vehicle & 77.71 & 103.30 & 74.05 & 101.64 & 133.66 & 230.83 & 127.22 & 230.48 & 426.44 & 460.07 & 414.48 & 453.87 \\ 
  waveform & 123.77 & 120.01 & 120.99 & 113.94 & 199.00 & 203.38 & 191.77 & 209.61 & 310.84 & 499.05 & 291.96 & 478.59 \\ 
  segment & 114.91 & 104.77 & 95.63 & 91.31 & 281.11 & 294.17 & 270.63 & 355.46 & 742.63 & 844.11 & 752.74 & 814.62 \\ 
  dna & 217.99 & 238.48 & 63.66 & 67.26 & 481.65 & 736.58 & 505.59 & 772.32 & 911.68 & 986.60 & 927.30 & 991.63 \\ 
  satimage & 53.51 & 66.19 & 54.04 & 66.33 & 82.50 & 90.50 & 81.87 & 90.11 & 160.84 & 168.85 & 156.83 & 157.31 \\ 
   \hline
\end{tabular}
}
}
\caption{Comparison of the ESS of the conditional class probability between Bayes MLR and paSB-MLR.}\label{ess}
\end{table}

We further measure how well  the Gibbs sampler is mixing using effective sample size (ESS) for both paSB-MLR and Bayesian multinomial logistic regression (Bayes MLR) of \citet{polson2013bayesian}. For both algorithms we let 
$
\betav_j\sim\mathcal{N}\left(0,\mbox{diag}(\alpha_{j0}^{-1},\ldots,\alpha_{jV}^{-1})\right)$, where $\alpha_{jv}\sim\mbox{Gamma}(0.001,1/0.001). 
$
   The ESS \citep{holmes2006bayesian} of a parameter or a function of parameters 
  is defined as 
$
\mbox{ESS}=L/\left[ 1+ 2 \sum_{h=1}^\infty \rho(h)\right],
$ 
where $L$ is the number of post-burn-in samples, $\rho(h)$ is the $h$th autocorrelation of the parameter or the function of parameters. It describes how quickly an MCMC algorithm generates independent samples. 
Since  the Gibbs sampler of Bayes MLR samples one $\betav_j$ conditioning on all $\betav_{j'}$ for $j'\neq j$, 
  which may lead to strong dependencies between different categories and hence slow down the mixing of the Markov chain. By contrast, 
  the 
  $\betav_{j}$'s are conditionally independent given the augmented variables $b_{ij}$'s in paSB-MLR, which may lead to faster mixing.  
  For both Bayes MLR and paSB-MLR, we consider five independent random trials, in each of which we randomly initialize the model parameters, run 10,000 Gibbs sampling iterations, and collect the last 1,000 MCMC samples of  $\betav_{j}$. 
We use the \texttt{mcmcse} package \citep{mcmcse} to estimate the ESS of each $p_{ij}$ in a random trial using the 1,000 collected MCMC samples. For the training set, we calculate the 10\% quantile, median, and 90\% quantile of the ESSs of all $p_{ij}$ for each random trial, and then report their averages over the five random trials in  Table \ref{ess}. For the testing set, we follow the same steps and report the results in Table \ref{ess}.
While paSB-MLR underperforms Bayes MLR on some of the data sets for the 10\% ESS quantile, they consistently outperform Bayes MLR on all data sets for both the ESS median and $90\%$ ESS quantile, for both training and testing.


\begin{figure}[t]
\centering
 \centering
 \includegraphics[width=.45\columnwidth]{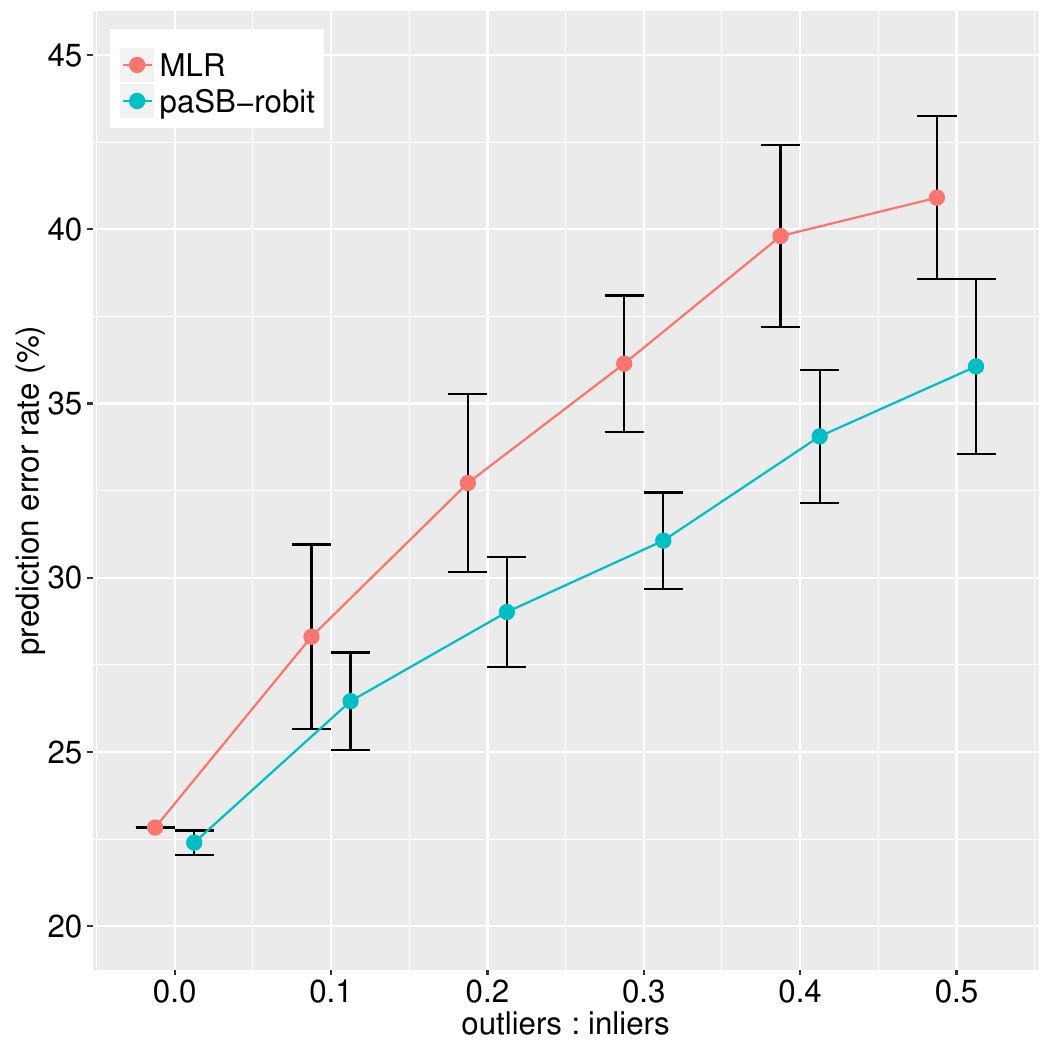} 
 \caption{Prediction error rates (\%, mean $\pm$ standard deviation) for different ratios of outliers to inliers.} 
 \label{fig:robust}
\end{figure}

\subsection{Robustness of paSB-Robit Regression}\label{sec:robust}
We use the contaminated vehicle data to demonstrate the robustness of paSB-robit. 
As discussed by \citet{liu2004robit}, the heavy-tailed conditional class probability function of robit regression can robustify the decision boundary when there exist outliers. We use the  vehicle training set as inliers, synthesize outliers that are far from inliers,  combine both as the new training set, and keep the testing set unchanged. We generate different numbers of outliers so that the ratio of outliers to inliers varies from $0$, $0.1$, $0.2$, $0.3$, to $0.5$, at each of which we randomly simulate 10 different sets of outlier covariates. We provide  the details on how we generate outliers in the Appendix.

We compare $L_2$-MLR and paSB-robit with $\kappa=1$ degree of freedom  on the contaminated vehicle data. Figure \ref{fig:robust} shows the prediction error rate (mean $\pm$ standard deviation) of the testing set for different outlier-inlier ratios. When there are no outliers, both approaches delivers comparable performances. As the ratio increases, paSB-robit with $\kappa=1$ more and more clearly outperforms $L_2$-MLR, 
which justifies the robustness of paSB-robit.

\section{Conclusions}\label{sec:conclude}

To transform a cross-entropy-loss binary classifier into a Bayesian multinomial regression model and derive efficient Bayesian inference, we develop a permuted and augmented stick-breaking construction. With permutation, we one-to-one map the categories to sticks to escape from poor category-stick mappings that impose restrictive geometric constraints on the decision boundaries, and with augmentation, we link a category outcome to conditionally independent stick-specific covariate-dependent Bernoulli random variables. 
 We illustrate this general framework by extending binary softplus regression, robit regression, and support vector machine into multinomial ones. Experiment results validate our contributions and show that the proposed multinomial softplus regressions achieve a good compromise between interpretability, complexity, and predictability. 

\subsection*{Acknowledgments}
The authors would like to thank the editor and two anonymous referees for their insightful
and constructive comments and suggestions, and Texas Advanced Computing Center for computational support.

%


\appendix

\section{Additional Lemma and Proofs}

\begin{proof}[Proof of Theorem \ref{thm1}]
The conditional probability of $y_i$ given $\{z_s,\pi_{is}\}_{1,S}$ can be expressed as
\begin{align*} 
\footnotesize
 \textstyle
P(y_i = s\given \{z_s,\pi_{is}\}_{1,S}) & =\textstyle\sum_{b_{ij}:j>z_s}\textstyle \left[P(b_{iz_s}=1)\right]^{\mathbf{1}(z_s\neq S)} \left[\prod_{j<z_s}P(b_{ij}=0)\right]\left[ \prod_{j>z_s}P(b_{ij}) \right]\notag\\
 &=\textstyle\left[P(b_{iz_s}=1)\right]^{\mathbf{1}(z_s\neq S)}\left[\prod_{j<z_s}P(b_{ij}=0)\right] \sum_{b_{ij}:j>z_s} \left[ \prod_{j>z_s}P(b_{ij}) \right], \notag 
 \end{align*}
which becomes the same as \eqref{eq:p_is} by applying \eqref{eq:AugSB} and 
 $\sum_{b_{ij}:j>z_s} \!\big[ \prod_{j>z_s}\! P(b_{ij}) \big]
 \!=\!1$. 
\end{proof}\vspace{-2.5mm}

\begin{proof}[Proof of Lemma \ref{lem_iia}] Under the paSB construction, the probability ratio of categories (choices) $s$ and $s+d$ is a function of the stick success probabilities $\pi_{z_{s}}, \pi_{z_{(s+1)}},\cdots, \pi_{z_{(s+d)}}$. More specifically,
\beq
 \frac{p_{i(s+d)}(\zv)}{p_{is}(\zv)}= 
 \frac {
 \pi_{iz_{(s+d)}}^{\mathbf{1}(z_{(s+d)}\neq S)} 
 \left[\prod_{z_s\le j <z_{(s+d)}}(1-\pi_{ij})\right]^{\delta(z_s\le z_{(s+d)})} 
 } 
 { 
 \pi_{iz_s}^{\mathbf{1}(z_s\neq S)}
 \left[\prod_{z_{(s+d)}\le j < z_s}(1-\pi_{ij})\right]^{\delta(z_s > z_{(s+d)})}
 }. \notag
\eeq
\end{proof}

\begin{proof}[Proof of Lemma \ref{lem:K1T1}]
Since $
p_{is}(\zv) = \left[1-\big({1+e^{\xv_i'\betav_s}}\big)^{-r_s}\right]^{\mathbf{1}(s\neq S)} \prod_{j<s} \big({1+e^{\xv_i'\betav_j}}\big)^{-r_j} 
$ when $K=T=1$ and $ \zv=(1,\ldots,S)$, 
the set of solutions to $p_{is}>p_0$ are bounded by the set of solutions to 
$
\big({1+e^{\xv_i'\betav_j}}\big)^{-r_s}>p_0,~j\in\{1,\ldots,s-1\},
$
and
$
1-\big({1+e^{\xv_i'\betav_s}}\big)^{-r_s}>p_0,
$
and hence bounded by the convex polytope defined by the set of solutions to the $s$ inequalities as
\beq 
\xv_i' [(-1)^{\mathbf{1}(j=s)}\betav_{j}]< (-1)^{\mathbf{1}(j=s)} \ln\left\{\left[p_0^{\mathbf{1}(j\neq s)}(1-p_0)^{\mathbf{1}(j= s)}\right]^{-\frac{1}{r_{j}}}-1\right\}, ~~ j\in\{1,\ldots,s\}. \notag
\eeq
\end{proof}

\begin{lem}\label{thm_latent_u}
Without loss of generality, let us assume that the category-stick mapping is fixed at $\zv=(1,2,\cdots,S)$.
The paSB multinomial logistic model that assigns choice $s\in\{1,\ldots,S\}$ for individual $i$ with probability $p_{is}= (\pi_{is})^{\mathbf{1}(s\neq S)}\prod_{j<s}(1-\pi_{ij})$, where $\pi_{is}=1/(1+e^{-W_{is}})$, can be considered as a sequential random utility maximization model. This model selects choice $s$ once $U_{is}>\sum_{j\ge s} U_{ij}$ is observed for $s=1,\ldots,S$, where $U_{is}$ are defined as
\begin{align*}
U_{i1}&=U_{i2}+\cdots+U_{iS}+W_{i1}+\varepsilon_{i1},\\
&\cdots\\
U_{is}&=\sum_{j>s}U_{ij}+ W_{is}+\varepsilon_{is},\\
&\cdots\\
U_{i(S-1)}&=W_{i(S-1)}+\varepsilon_{i(S-1)},\\
U_{iS}&=0,
\end{align*}
and $\varepsilon_{is}\overset{i.i.d.}{\sim} \emph{\mbox{Logistic}}(0,1)$ are independent, and identically distributed (i.i.d.) random variables following the standard logistic distribution. 
\end{lem}

\begin{proof}[Proof of Lemma \ref{thm_latent_u}]
Note that $P(\varepsilon<x)=1/(1+e^{-x})$ if $\varepsilon\sim\mbox{Logistic}(0,1)$. First consider the choice of individual $i$ be $y_i=1$, which would happen with probability 
\begin{align*}
P(y_i=1)=P\left(U_{i1}>\sum \nolimits_{j\geq 1} U_{ij}\right)=P(\varepsilon_{i1}>-W_{i1})=1/(1+e^{-W_{i1}})=\pi_{i1}=p_{i1}.
\end{align*} 
Then for $s=2,\cdots, S-1$,
\begin{align*}
P(y_i=s)&=P(y_i=s\given y_i>s-1)P(y_i>s-1)\\
&= P\left(U_{is}>\sum \nolimits_{j>s}U_{ij}\right) \prod\nolimits_{j\leq s-1} P\left(U_{ij}<\sum \nolimits_{j'>j}U_{ij'}\right)\\
&= P(\varepsilon_{is}>-W_{is}) \prod\nolimits_{j\leq s-1} P(\varepsilon_{ij}<-W_{ij})\\
&= \pi_{is}\prod\nolimits_{j\leq s-1}(1-\pi_{ij})\\
&=p_{is}.
\end{align*}
Finally, $P(y_i=S)=1-\sum_{j<S} P(y_i=j)=\prod_{j<S}(1-\pi_{ij})=p_{iS}$.
\end{proof}

\section{Experimental Settings and Additional Results}
The table below summarizes the sizes of both training and testing sets, and the number of covariates and categories. The training and testing sets are predefined for vehicle, 
dna, and satimage. Note that the training and validation sets are combined as training. We divide the other data sets into training and testing as follows. For iris, wine, and glass, five random partitions are taken such that for each partition the training set accounted for 80\% of the whole data set  while the testing set 20\%. The classification error rate is calculated by averaging the error rates of all five random partitions. For square, waveform, and segment, only one random partition is taken, where 70\% of the square data set are used as training and the remaining 30\% as testing, and 10\% of both the waveform and segment datas are used as training and the remaining 90\% as testing.

\begin{table}[t]
\centering
\resizebox{\linewidth}{!}{%
\scriptsize
\begin{tabular}{l P{0.6cm}P{0.5cm}P{0.5cm}P{0.5cm}P{0.7cm}P{1cm}P{1cm}
P{0.5cm}P{1cm}}
 \hline
 & square & iris & wine & glass & vehicle & waveform & segment & 
 dna & satimage \\ 
 \hline
Train size & 294 & 120 & 142 & 171 & 592 & 500 & 231 & 
2000 & 4435 \\ 
 Test size & 126 & 30 & 36 & 43 & 254 & 4500 & 2079 & 
 1186 & 2000 \\ 
 Covariate number& 2 & 4 & 13 & 9 & 18 & 21 & 19 & 
 180 & 36 \\ 
 Category number & 3 & 3 & 3 & 6 & 4 & 3 & 7 & 
 3 & 6 \\ 
 \hline
\end{tabular}
}
\caption{Multi-class classification data sets used in experiments for model comparison.}
\end{table}

We compare paSB-MLR, paSB-robit with $\kappa=6$, paSB-MSVM, and MSR with three other models, including $L_2$ regularized multinomial logistic regression ($L_2$-MLR), support vector machine (SVM), and adaptive multi-hyperplane machine (AMM). For paSB-robit, we run 8,000 iterations and discard the first 5,000 as burn-in (this setting is unchanged for experiments in Section \ref{sec:robust}). For paSB-MSVM, we use the spike-and-slab prior to select the kernel bases and set $0.5$ as the probability of spike at $0$, which is referred to as a uniform prior by \citet{polson2011data}. A Gaussian radial basis function (RBF) kernel is used and the kernel width is selected by 3-fold cross validation from $(2^{-10}, 2^{-9},\ldots, 2^{10})$. We run 1000 MCMC iterations and discard the first 500 as burn-in samples. For MSR, we try both paSB and parSB with $(K,T)$ set as $(1, 1)$, $(1, 3)$, $(5, 1)$, or $(5, 3)$. We run 10000 MCMC iterations and discard the first 5000 as burn-in samples. The predictive probability is calculated by averaging the Monte Carlo average predictive probabilities from paSB and parSB MSRs. An observation in the testing set is classified to the category associated with the largest predictive probability.

\begin{table}[t]
\centering
\makebox[\linewidth]{
\resizebox{\linewidth}{!}{%
\begin{tabular}{l P{1.5cm}P{1.5cm}P{1.8cm}P{1.5cm}P{1.6cm}P{1.6cm}P{1.6cm}P{1.6cm}P{1.6cm}P{1.6cm}P{1.6cm}}
 \hline
& paSB-MLR & paSB-robit& paSB-MSVM & $K=1$ $T=1$ & $K=1$ $T=3$ & $K=5$ $T=1$ & $K=5$ $T=3$ & DT-MSR & $L_2$-MLR & SVM & AMM \\ 
 \hline
square  & 53.17(2) & 57.94(2)& 0(256) & 57.14(4) & 13.49(12) & 1.59(10) & 0.79(33) &1.59(40) & 62.29(2) & 4.76(22) & 16.67(7) \\ 
  iris & 2(2) & 6.67(2)& 3.33(97.6) & 4.67(4) & 4.67(12) & 4(5.4) & 3.33(12) & 4.67(16.2)& 3.33(2) & 4(35) & 4.67(8.6) \\ 
  wine  & 8.33(2)&4.86(2) & 2.14(125.8) & 4.45(4) & 4.45(12) & 6.67(4) & 3.34(12)&3.89(16.4) & 3.89(2) & 2.78(77.2) & 3.89(7.8) \\ 
  glass & 39.07(5) &34.41(5)& 30.23(137.6) & 35.35(10) & 30.7(30) & 33.02(10.4) & 35.81(33.2) &32.56(47.8) &33.02(5) & 28.84(118) & 37.67(23.8) \\ 
  vehicle& 25.98(3)& 22.44(3) & 17.71(592) & 23.62(6) & 21.65(18) & 18.9(12) & 16.93(33) &18.11(45) & 22.83(3) & 18.50(256) & 21.89(17) \\ 
  waveform& 18.78(2) & 16.84(2)& 16.56(500) & 19.73(4) & 17.11(12) & 17.07(6) & 17.11(18) &16.49(23) & 15.60(2) & 15.22(212) & 18.54(11.6) \\ 
  segment & 7.07(6) &9.81(6)& 9.86(231) & 8.61(12) & 8.37(36) & 7.31(13) & 7.79(36)&8.85(50) & 8.56(6) & 6.20(93) & 12.47(11.4) \\ 
  dna & 6.58(2)&4.30(2)  & 7.25(1701) & 5.56(4) & 5.73(12) & 6.07(7) & 5.82(12) &4.72(17) & 5.98(2) & 4.97(1142) & 5.43(18.6) \\ 
  satimage  & 21.35(5)&16.40(5) & 9.5(4315) & 15.8(10) & 14.7(30) & 13.25(34) & 11.95(102) &11.55(105) & 17.80(5) & 8.50(1652) & 15.31(16.8) \\ 
  \hline
\end{tabular}
}
}\caption{Comparison of classification error rates (\%) of paSB-MLR, paSB-robit,  paSB-MSVM, MSRs with various $K$ and $T$ (results column 3 to 6), MSR with data transformation (DT-MSR), $L_2$-MLR, SVM, and AMM, using the collected MCMC sample with the highest log-likelihood. The number of active hyperplanes/support vectors used for out-of-sample predictions are shown in parenthesis.}\label{compare1}
\end{table}

 \begin{figure}[!t]
\centering
 \centering
 \includegraphics[width=.52\columnwidth]{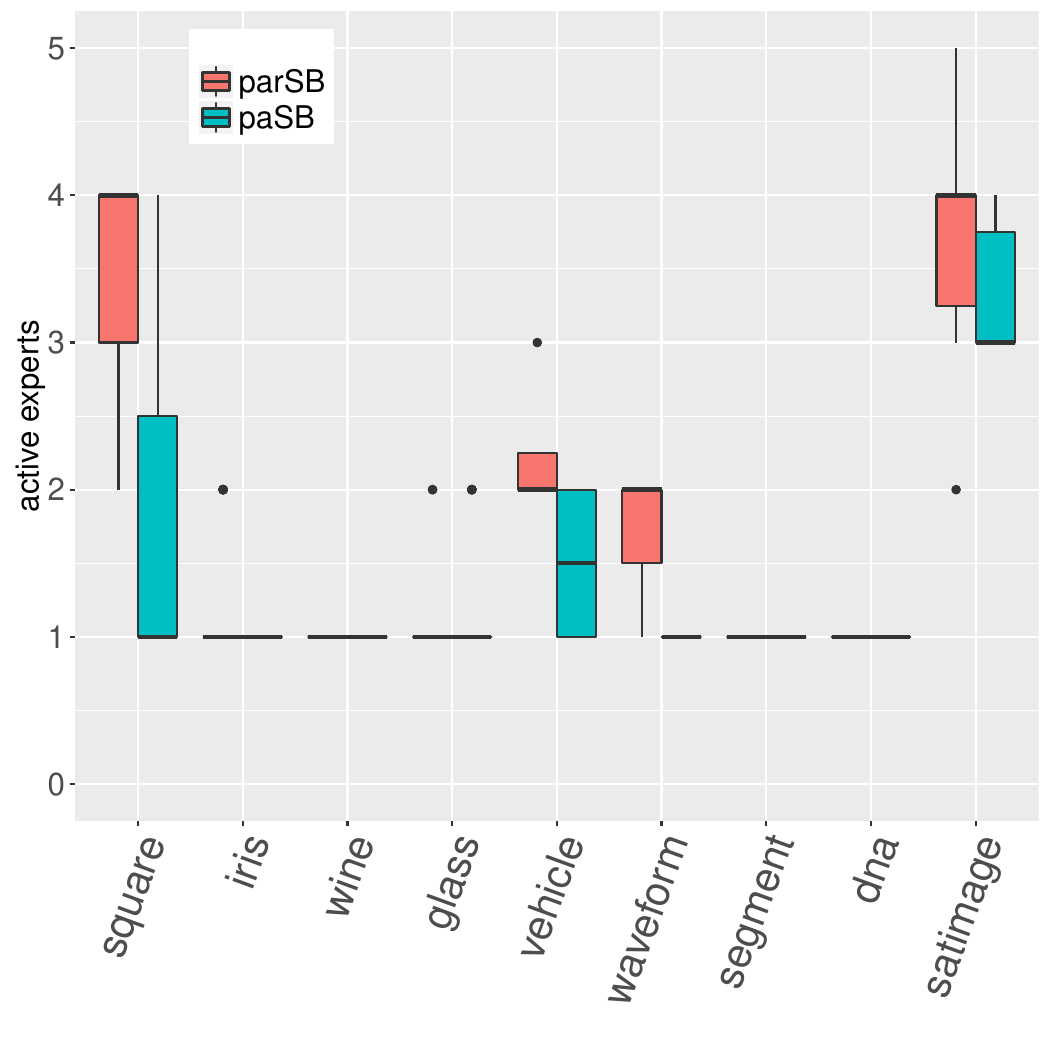}
 \caption{Boxplots of the number of active experts inferred by paSB/parSB MSRs with $K=5$ and $T=3$.} 
 \label{boxplot}
\end{figure}

For MSRs on benchmark data with data transformation, in the first step, we run paSB-MSR and parSB-MSR with $K=5$ and $T=3$ on the original covariates for 3,000 iterations 
to learn $\zv$, $\{r_{jk}\}$, $\{\betav_{jk}^{(t)}\}$. 
We then transform the covariates by Equation \eqref{data_trans}, where $\betav_{jk}^{(t)}$ are associated with active experts. In the last step, we use $\zv$ learned in the first step, and run paSB-MSR and parSB-MSR with $K=5$ and $T=1$ for 10,000 iterations and collect the last 5,000 samples to compute the predictive probabilities. 

We use the $L_2$-MLR provided in the LIBLINEAR package \citep{REF08a} to train a linear classifier, where a bias term is included and the regularization parameter 
 is five-fold cross-validated on the training set from $(2^{-10}, 2^{-9},\ldots, 2^{15})$. 
We run the LIBLINEAR package in R via \texttt{R} package \texttt{LiblineaR} \citep{LiblineaR}.
We also classify an observation to the category associated with the largest predictive probability. For SVM, we use the LIBSVM package \citep{LIBSVM} and run it in R  with \texttt{R} package \texttt{e1071} \citep{e1071}.  
A Gaussian RBF kernel is used and three-fold cross validation is adopted to tune both the regularization parameter 
and kernel width from $(2^{-10}, 2^{-9},\ldots, 2^{10})$ 
on the training set. For paSB-MSVM, we use three-fold cross validation on the training set to select a kernel width from $(2^{-10}, 2^{-9},\ldots, 2^{10})$. 
We choose the default LIBSVM settings for all the other parameters. 
We consider adaptive multi-hyperplane machine (AMM) of \citet{wang2011trading}, as implemented in the BudgetSVM\footnote{\url{http://www.dabi.temple.edu/budgetedsvm/}} (Version 1.1) software package \citep{BudgetSVM}. We use the batch version of the algorithm. Important parameters of the AMM include both the regularization parameter $\nu$ and training epochs $E$. As also mentioned by \citet{kantchelian2014large}, we do not observe the testing errors of AMM to strictly decrease as $E$ increased. Thus, in addition to cross validating the regularization parameter $\nu$ on the training set from $\{10^{-7}, 10^{-6},\ldots, 10^{-2}\}$, as done in \cite{wang2011trading}, for each $\nu$, we try $E\in\{5,10,20,50,100\}$ sequentially until the cross-validation error begins to decrease, $i.e.$, under the same $\nu$, we choose $E=20$ if the cross-validation error of $E=50$ is greater than that of $E=20$. We use the default settings for all the other parameters, and calculate average classification error rates.

We add an outlier to the vehicle data in Section \ref{sec:robust} as follows. 
There are 18 covariates, whose values range from $-1$ to $1$, in this data set. 
To simulate an outlier, since the MLR regression coefficients associated with the 4th, 5th, and 6th covariates all have large absolute values,  we first draw three uniform random numbers from $(-3,-2)\cup (2,3)$ and assign them to these three covariates, and then assign each of the 12 remaining covariates  a  uniform random number from $(-1,1)$. 
Finally, we draw the category label $y$ uniformly from $\{1,2,3,4\}$. 
\bibliographystyle{abbrvnat}


\end{document}